\documentclass[11pt]{article}
\usepackage{indentfirst}
\usepackage[margin=1in]{geometry}
\usepackage{graphicx}
\usepackage{amsmath}
\usepackage{xcolor}
\usepackage{amssymb}
\usepackage{amsthm}
\usepackage{tikz}
\usetikzlibrary{positioning}

\newtheorem{theorem}{Theorem}
\newtheorem{corollary}{Corollary}
\newtheorem{proposition}{Proposition}
\newtheorem{remark}{Remark}
\newtheorem{example}{Example}

\title{Why Self-Supervised Encoders Want to Be Normal}
\author{Yuval Domb\\\texttt{yuval@moonmath.ai}}
\date{May 03, 2026}

\begin{document}

\maketitle

\begin{abstract}
Self-supervised learning has achieved remarkable empirical success in learning robust representations without explicit labels, most recently demonstrated within the framework of Joint-Embedding Predictive Architectures (JEPA).
However, a fundamental question remains: what analytical principles drive these encoders toward specific distributional states? In this paper, we demonstrate that the preference for normal distributions in self-supervised encoders is a direct consequence of the Information Bottleneck (IB) principle. By recasting the IB objective as a rate-distortion problem over the predictive manifold, we provide a theoretical basis for why optimal, target-neutral, latent representations should tend towards isotropic Gaussian states.

Under this framework, we show that latent representations correspond to soft clustering of inputs sharing similar predictive distributions, organized within a natural simplex structure. This perspective unifies a wide range of existing supervised and less-supervised objectives and provides a principled explanation for commonly used regularization schemes. Furthermore, we derive practical loss objectives that approximate this structure and demonstrate their effectiveness on standard benchmarks. Ultimately, our framework offers a geometric lens to understanding representation collapse and it establishes a mathematical system for regularization strategies to be used to ensure high-entropy, informative embeddings in modern self-supervised models.

\end{abstract}

%%%%%%%%%%%%%%%%%%%%%%%%%%%%%%%%%%%%%%%%%%%%%%%%%%%%%%%%%%%%

\section{Introduction}

Learning compact representations that preserve predictive information
is a central problem in machine learning.
Given observations $X$ and targets $Y$, the goal is to learn a latent
variable $W$ that compresses $X$ while retaining enough information to
predict $Y$.
The IB method \cite{tishby1999information}
formalized this tradeoff by minimizing the mutual information $I(X;W)$
subject to a constraint on the predictive information $I(W;Y)$.

This paper develops a geometric and information-theoretic
framework for encoder-decoder learning, built on the IB
principle.
Starting from the rate-distortion equivalence of Harrem\"oes and
Tishby \cite{harremoes2007information}, with Kullback-Leibler (KL) divergence as the
distortion measure, we characterize the optimal representation
geometrically as a soft clustering of the \textit{predictive manifold}
$\mathcal{M}=\{p(Y|x):x\in\mathcal{X}\}$%
\footnote{We use the term ``manifold'' loosely. $\mathcal{M}$ is the
image of $x\mapsto p(Y|x)$ and may have singularities,
self-intersections, or varying dimension.}.
A key outcome is a theoretical justification of the Sketched Isotropic Gaussian Regularization (SIGReg) method
\cite{balestriero2025lejepa} as a relaxation of the maximum entropy
principle on the predictive manifold, mediated by a chain of
transformations from the flat Dirichlet distribution to the isotropic
Gaussian.
SIGReg, together with the Conditional Entropy Bottleneck (CEB) decomposition of Fischer
\cite{fischer2020conditional}, yields concrete encoder losses for
supervised and semi-supervised settings. In the self-supervised
setting, SIGReg serves as the distributional regularizer alongside a
view-prediction proxy.

\begin{figure}[h]
\centering
\begin{tikzpicture}[>=stealth]

% ===== Panel (a): Optimal encoder =====
\begin{scope}[shift={(-0.5,0)}]
    \node[draw, circle, fill=red!10, inner sep=0pt, minimum size=8mm,
          font=\small] (Y) at (0,0) {$Y$};
    \node[draw, circle, fill=blue!10, inner sep=0pt, minimum size=8mm,
          font=\small] (X) at (1.5,0) {$X$};
    \node[draw, circle, fill=green!8, inner sep=0pt, minimum size=8mm,
          font=\small] (W) at (3.0,0) {$W$};

    \draw[thick] (Y) -- (X);
    \draw[thick, ->] (X) -- (W)
        node[midway, above, font=\scriptsize] {enc};

    \draw[<->, blue!60!black, semithick] (1.5,-0.7) -- (3.0,-0.7);
    \node[font=\scriptsize, blue!60!black] at (2.25,-1.0)
        {min $I(X;W)$};

    \draw[<->, red!60!black, semithick] (0,0.7) -- (3.0,0.7);
    \node[font=\scriptsize, red!60!black] at (1.5,1.0)
        {max $I(W;Y)$};

    % \node[font=\small\bfseries, align=center] at (1.5,-1.7)
    %     {(a) IB optimal\\encoding};
\end{scope}

% ===== Arrow (a)->(b) =====
\draw[->, very thick, gray!40] (3.6,0) -- (4.2,0);

% ===== Panel (b): Soft clustering =====
\begin{scope}[shift={(4.5,-1.3)}]
    \draw[thick] (0,0) -- (3.4,0) -- (1.7,2.94) -- cycle;

    % Cluster 1: compact, bottom-left
    \draw[blue!35, fill=blue!10] (0.7,0.45) ellipse (0.28 and 0.38);
    \foreach \x/\y in {0.62/0.25, 0.78/0.42, 0.65/0.58,
                       0.82/0.65, 0.58/0.45, 0.75/0.30}
        \fill[blue!50] (\x,\y) circle (1.5pt);
    \fill[black] (0.70,0.44) circle (2.5pt);

    % Cluster 2: elongated, top-center
    \draw[blue!35, fill=blue!10] (1.55,1.55) ellipse (0.45 and 0.25);
    \foreach \x/\y in {1.25/1.50, 1.45/1.62, 1.65/1.48,
                       1.80/1.58, 1.50/1.42, 1.38/1.58}
        \fill[blue!50] (\x,\y) circle (1.5pt);
    \fill[black] (1.53,1.53) circle (2.5pt);

    % Cluster 3: wide, bottom-right
    \draw[blue!35, fill=blue!10] (2.55,0.55) ellipse (0.42 and 0.3);
    \foreach \x/\y in {2.30/0.45, 2.50/0.65, 2.72/0.50,
                       2.85/0.62, 2.45/0.40, 2.68/0.70}
        \fill[blue!50] (\x,\y) circle (1.5pt);
    \fill[black] (2.55,0.54) circle (2.5pt);

    \node[font=\scriptsize] at (1.7,1.0)
        {$W\!\in\!\mathcal{P}_K$};
    % \node[font=\small\bfseries, align=center] at (1.7,-0.75)
    %     {(b) Soft clustering\\of $p(Y|X)$};
\end{scope}

% ===== Arrow (b)->(c) =====
\draw[->, very thick, gray!40] (8.4,0) -- (9.0,0);

% ===== Panel (c): Max entropy latent =====
\begin{scope}[shift={(9.3,-0.6)}]
    \draw[blue!15, fill=blue!3]  (1.5,0.5) circle (1.15);
    \draw[blue!25, fill=blue!5]  (1.5,0.5) circle (0.75);
    \draw[blue!40, fill=blue!8]  (1.5,0.5) circle (0.35);

    \foreach \x/\y in {1.50/0.50, 1.28/0.83, 1.72/0.17,
                       1.08/0.38, 1.92/0.62, 1.58/1.10,
                       1.42/-0.10, 0.70/0.50, 2.30/0.50,
                       1.50/1.25, 1.10/0.05, 1.90/0.95,
                       0.82/0.68, 2.18/0.32, 1.62/0.275,
                       1.38/0.725, 1.50/0.08, 1.05/0.92,
                       1.85/0.62, 1.72/0.50, 1.18/0.50,
                       1.40/0.98, 1.68/0.77, 1.55/0.575}
        \fill[blue!60!black] (\x,\y) circle (0.9pt);

    \node[font=\scriptsize] at (1.5,2.0)
        {$W\!\sim\!\mathcal{N}(0,I)$};
    % \node[font=\small\bfseries, align=center] at (1.5,-0.95)
    %     {(c) Max entropy $W$\\for less supervision};
\end{scope}

\end{tikzpicture}
\caption{Overview from  left to right.
(a)~The IB optimal encoder.
(b)~Soft clustering in the probability simplex.
(c)~Max entropy for less supervised settings.}
\label{fig:overview}
\end{figure}
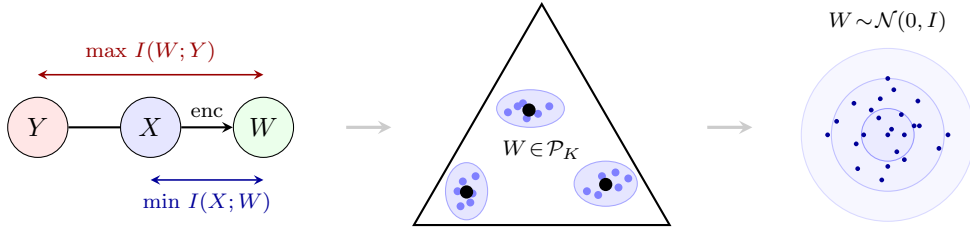

The framework applies to both discrete and continuous variables.
For discrete $Y$ with $K$ outcomes, the predictive manifold lives in
the probability simplex $\mathcal{P}_K$.
For continuous $Y$, we introduce the \textit{effective predictive
dimension} and show that the manifold has finite effective dimension
bounded by the intrinsic dimension of $X$.

\subsection{Contributions}

\begin{enumerate}
    \item \textbf{Geometric framework for optimal IB representations.}
    Building on the rate-distortion equivalence of
    Harrem\"oes and Tishby \cite{harremoes2007information} and the
    tutorial treatment of Goldfeld and Polyanskiy
    \cite{goldfeld2020information},     we assemble known ingredients,
    including the optimal encoder form, the canonical simplex
    representation, invariance under prediction-preserving
    transformations, and the dimensionality bound
    $\dim(W^*)\leq K{-}1$ (see Corollary~\ref{thm:dim}),
    into a unified geometric picture: the optimal encoder at any distortion level performs a
    soft clustering of the predictive manifold
    $\mathcal{M}\subset\mathcal{P}_K$, with a linear decoder in the
    canonical parameterization
    (see Section~\ref{sec:geometric}).
    We extend the dimensionality statement to continuous $Y$ via
    covering numbers (see Proposition~\ref{prop:geom_dim}).
    The individual results are not new. The contribution is the
    synthesis, which makes contributions~2 and~3 below concrete.

    \item \textbf{Dirichlet-to-Gaussian chain and characterization of
    SIGReg.}
    We derive a chain of exact transformations,
    $\mathrm{Dir}(1,\ldots,1)\to K$ i.i.d.\
    $\mathrm{Exp}(1/2)\to 2K$ i.i.d.\ $\mathcal{N}(0,1)$, that
    connects the maximum entropy prior on $\mathcal{P}_K$ to an
    isotropic Gaussian in Euclidean space, with quantified entropy
    overhead at each step
    (see Section~\ref{sec:simplex_to_euclidean}).
    This identifies SIGReg \cite{balestriero2025lejepa} as a
    Gaussian relaxation of the flat Dirichlet via the Cram\'er-Wold
    theorem, with a non-diminishing phase entropy overhead of
    $K\log 2\pi$ nats that affects rate accounting but not
    achievable prediction
    (see Section~\ref{sec:sig_connection}).

    \item \textbf{Non-parametric encoder losses across supervision
    regimes.}
    Using the CEB decomposition of Fischer
    \cite{fischer2020conditional}, which rewrites the IB Lagrangian as
    conditional rate $I(X;W|Y)$ plus total rate $I(X;W)$
    (eq.~\ref{eq:lagrangian_prime}), we derive encoder losses for
    the supervised and semi-supervised settings, estimated via
    minibatch marginals rather than variational bounds
    (Sections~\ref{sec:supervised}--\ref{sec:semi_supervised}).
    In the self-supervised setting, where no labels are available,
    the CEB conditional rate is replaced by a view-prediction proxy
    and serves as the distributional
    regularizer (see Section~\ref{sec:self_supervised}).
\end{enumerate}

\subsection{Prior Work}

\medskip
\noindent\textbf{The Information Bottleneck.}
The optimization problem underlying the IB method was first
formulated by Witsenhausen and Wyner
\cite{witsenhausen1975conditional}, who studied
$\inf_{p(w|x):\,H(X|W)\geq\alpha}H(Y|W)$, in the context of common
information \cite{gacs1973common} for discrete random
variables.
Tishby, Pereira, and Bialek \cite{tishby1999information}
independently introduced the same problem as a principled framework
for extracting relevant information, coining the term ``information
bottleneck.''
The optimal encoder satisfies self-consistent equations solvable by
the Blahut-Arimoto algorithm \cite{blahut1972computation,
arimoto1972algorithm}, originally developed for rate-distortion
computation.
The connection between IB and rate-distortion theory was made explicit
by Harrem\"oes and Tishby \cite{harremoes2007information}, who showed
that the IB is a rate-distortion function with KL~divergence as
distortion. We build on this equivalence and develop its geometric
consequences.

\medskip
\noindent\textbf{IB and deep learning.}
Goldfeld and Polyanskiy \cite{goldfeld2020information} provide a
comprehensive tutorial on IB theory and its applications to machine
learning, covering the discrete-variable formulation in detail.
Shwartz-Ziv and Tishby \cite{shwartzziv2017opening} studied the
information plane dynamics of deep networks during training.
Saxe et al.\ \cite{saxe2018information} showed that the reported
compression phase does not hold in general.
Achille and Soatto \cite{achille2018emergence} develop an IB-flavoured
account of representation learning in which invariance to nuisances
emerges from a variational IB objective. Our treatment is
complementary: non-variational, and with an explicit geometric
characterization of the optimum on the probability simplex.
Alemi et al.\ \cite{alemi2017deep} introduced VIB, using variational
bounds (a fixed prior to upper-bound $I(X;W)$ and a learned
decoder to lower-bound $I(W;Y)$) to make the IB objective tractable
for neural networks.
Fischer \cite{fischer2020conditional} introduced CEB, which rewrites
the IB Lagrangian as
$I(X;W|Y)-\gamma\,I(W;Y)$ via the chain rule of mutual information
(equivalent to our eq.~\eqref{eq:lagrangian_prime} with
$\gamma=\beta-1$), and demonstrated benefits for calibration and
adversarial robustness.
Fischer used variational bounds in practice. We estimate both CEB
terms via minibatch marginals rather than variational surrogates.
Subsequent work extended the variational IB approach to nonlinear
settings \cite{kolchinsky2019nonlinear} and connections to PAC-Bayes
generalization bounds \cite{xu2017information}.
Strouse and Schwab \cite{strouse2017deterministic} studied the
deterministic IB, which replaces the rate $I(X;W)$ with the marginal
entropy $H(W)$ and thereby admits deterministic encoders by
construction. The resulting frontier is in the $(H(W),I(W;Y))$ plane
rather than on the standard IB plane, i.e. $(I(X;W),I(W;Y))$.

\medskip
\noindent\textbf{Sufficient statistics and minimal representations.}
The connection between IB and sufficient statistics is classical:
the minimal sufficient statistic of $X$ for $Y$ achieves the endpoint
of the rate-distortion function at zero distortion
\cite{tishby1999information,cover2006elements}.
The concept of sufficiency originates with Fisher
\cite{fisher1922mathematical} and was formalized by the
Lehmann-Scheff\'e theorem \cite{lehmann1950completeness}.
Our Theorem~\ref{thm:mss} consolidates these classical facts with a
self-contained proof and restates the result in the geometric
language of the predictive manifold, serving as the foundation for
the soft-clustering interpretation at $\varepsilon>0$.

\medskip
\noindent\textbf{Geometric and information-geometric perspectives.}
Information geometry \cite{amari2016information} studies the
Riemannian structure of statistical manifolds equipped with the
Fisher metric.
Chechik et al.\ \cite{chechik2005information} analyzed the Gaussian
IB and its geometric properties.
Our geometric interpretation of the predictive manifold
$\mathcal{M}\subset\mathcal{P}_K$ and its soft clustering is
complementary: rather than studying the Fisher geometry of the model
family, we study the geometry of the set of predictive distributions
induced by the data.
The manifold hypothesis \cite{fefferman2016testing,pope2021intrinsic}
and effective-dimension constructions \cite{abbas2020effective}
provide context for our dimensionality bound. Our effective predictive
dimension (see Appendix~\ref{app:eff_dim}) is defined via covering
numbers and connects the intrinsic dimension of the input to the
complexity of $\mathcal{M}$.

\medskip
\noindent\textbf{Rate-distortion theory.}
The rate-distortion function, introduced by Shannon
\cite{shannon1959coding} and developed by Berger
\cite{berger1971rate}, characterizes the fundamental limits of lossy
compression.
Our use of KL~divergence as distortion connects to the extensive
literature on mismatched and logarithmic distortion measures
\cite{cover2006elements}.

\medskip
\noindent\textbf{Self-supervised learning and SIGReg.}
Self-supervised learning methods learn representations without
human labels, using pretext tasks or contrastive objectives
\cite{chen2020simple,grill2020bootstrap,caron2021emerging}.
Shwartz-Ziv and LeCun \cite{shwartzziv2024compress} review
self-supervised learning through an information-theoretic lens,
introducing a unified framework for the self-supervised
information-theoretic learning problem.
The Joint-Embedding Predictive Architecture (JEPA)
\cite{lecun2022path} proposes learning by predicting in
representation space.
LeJEPA \cite{balestriero2025lejepa} implements this paradigm with SIGReg, which
enforces an isotropic Gaussian prior on the encoder output via the
Cram\'er-Wold theorem and the Epps-Pulley test \cite{epps1983test}.
Our framework provides a theoretical foundation for this choice: SIGReg
implements a Gaussian relaxation of the maximum entropy principle on
the predictive manifold via the Dirichlet-to-Gaussian chain, with a
quantified entropy overhead that affects rate but not prediction.

\medskip
\noindent\textbf{Maximum entropy and the Dirichlet distribution.}
The maximum entropy principle \cite{jaynes1957information} selects the
least informative distribution consistent with known constraints.
On the probability simplex $\mathcal{P}_K$, the maximum entropy
distribution with respect to the Lebesgue measure is the flat Dirichlet
$\mathrm{Dir}(1,\ldots,1)$ \cite{cover2006elements}.
The Dirichlet family is widely used as a prior in Bayesian inference
\cite{ferguson1973bayesian} and in variational autoencoders on the
simplex \cite{joo2020dirichlet}. Our use is different: the flat
Dirichlet is not a learned prior but a fixed reference encoding
complete uncertainty about the predictive structure.

\medskip
\noindent\textbf{Semi-supervised learning.}
Semi-supervised methods leverage unlabeled data to improve learning
when labels are scarce
\cite{grandvalet2004semi,kingma2014semi,vanengelen2020survey}.
Kingma et al.\ \cite{kingma2014semi} extend variational autoencoders
to the semi-supervised setting. Grandvalet and Bengio
\cite{grandvalet2004semi} use entropy minimization on unlabeled data.
Our semi-supervised loss pairs the CEB conditional rate on labeled
data with SIGReg on all data, replacing task-specific assumptions
with a distributional regularizer grounded in the maximum entropy
principle.

%%%%%%%%%%%%%%%%%%%%%%%%%%%%%%%%%%%%%%%%%%%%%%%%%%%%%%%%%%%%

\section{Problem Statement}

Given a random variable pair $(X,Y)\sim p(x,y)$, the goal is to
approximate $p(y|x)$ via a learned encoder-decoder network, using
i.i.d.\ samples from $(X,Y)$.

\subsection{Preliminaries}
\label{sec:prelim}

$X$ and $Y$ take values in measurable spaces $\mathcal{X}$ and
$\mathcal{Y}$, respectively, and may be discrete or continuous.
We write sums throughout. Continuous variables are handled by replacing
sums with integrals and Shannon entropy with differential entropy,
under standard regularity conditions.

We use standard information-theoretic notation
\cite{cover2006elements}: $H(X)$ denotes the entropy of $X$,
$H(Y|X)$ the conditional entropy,
$I(X;Y)=H(Y)-H(Y|X)$ the mutual information,
$I(X;Y|W)$ the conditional mutual information given~$W$, and
$D_{\mathrm{KL}}(p\|q)=\sum_x p(x)\log\frac{p(x)}{q(x)}$ the
KL divergence between distributions $p$ and $q$.
All logarithms are natural unless stated otherwise.

We denote by
\begin{align}
    \mathcal{P}_K := \left\{ \mathbf{p} \in \mathbb{R}^K : p_i \geq 0,\;
    \sum_{i=1}^K p_i = 1 \right\}
    \label{eq:simplex}
\end{align}
the $(K{-}1)$-dimensional probability simplex.
Here $K$ is the \textit{effective predictive dimension}: the dimension
of the smallest simplex containing the predictive manifold
$\mathcal{M}:=\{p(Y|x):x\in\mathcal{X}\}$.
\begin{itemize}
    \item When $Y$ is discrete then
    $K=|\mathcal{Y}|$, i.e. $p(Y|x)$ is a $K$-dimensional
    probability vector.
    \item When $Y$ is continuous, the space of densities is
    infinite-dimensional, but the predictive manifold $\mathcal{M}$
    has finite effective dimension, and we set
    $K=\lceil d_{\mathrm{eff}}(\mathcal{M}) \rceil + 1$,
    bounded by $d_X$, the intrinsic dimension of $X$, under Lipschitz
    regularity of $f^*\colon x\mapsto p(Y|x)$ (see Appendix~\ref{app:eff_dim}).
    $\mathcal{P}_K$ is then the $(K{-}1)$-simplex spanned by $K$
    densities whose convex hull contains $\mathcal{M}$
    (equivalently, a sufficient $K$-bin quantization of $\mathcal{Y}$
    adapted to $\mathcal{M}$), and all subsequent geometric
    statements refer to this chart.
\end{itemize}
In both cases, the optimal representation lives in $\mathcal{P}_K$
and all subsequent results are stated in terms of~$K$.

Following the information bottleneck framework
\cite{tishby1999information}, we partition the problem into an
encoder-decoder pair: the encoder maps $X\rightarrow W$ and the decoder
maps $W\rightarrow \hat{Y}$, where $\hat{Y}$ is sampled from $p(y|w)$
and $W$ is the latent variable.
The system obeys the Markov chain
\begin{align}
    Y - X - W - \hat{Y},
    \label{eq:markov}
\end{align}
which implies $p(w|x,y)=p(w|x)$: the encoder has no direct access
to $Y$.

\subsection{Optimization Problem}

By the chain rule for mutual information,
$I(X,W;Y)=I(X;Y)+I(W;Y|X)=I(W;Y)+I(X;Y|W)$.
Since $I(W;Y|X)=0$ by \eqref{eq:markov}, we obtain
\begin{align}
    I(X;Y) - I(W;Y) = I(X;Y|W).
    \label{eq:identity}
\end{align}
where $I(X;Y|W)$ is the residual predictive information between
$X$ and $Y$ not captured by~$W$.

The goal is to learn an encoder that minimizes the encoding rate
$R=I(X;W)$ subject to an upper bound on this residual:
\begin{align}
    \min_{p(w|x)} \; I(X;W) \quad \text{s.t.} \quad
    I(X;Y|W) \leq \varepsilon.
    \label{eq:opt}
\end{align}
By the data processing inequality \cite{cover2006elements} applied to
\eqref{eq:markov}, $I(W;Y)\leq I(X;Y)$, so $I(X;Y|W)\geq 0$ and the
constraint is always satisfiable.
The parameter $0\leq \varepsilon\leq I(X;Y)$ controls the compression-prediction
tradeoff:
\begin{itemize}
    \item $\varepsilon=0$ enforces exact sufficiency, i.e.\
    $I(W;Y)=I(X;Y)$, recovering the minimal sufficient statistic of
    $X$ with respect to $Y$
    \cite{fisher1922mathematical,lehmann1950completeness}.
    \item $\varepsilon>0$ allows approximate sufficiency, permitting
    greater compression at the cost of some predictive information loss.
\end{itemize}

\subsection{Lagrangian Formulation}
\label{sec:lagrangian}

Introducing $\beta\geq 0$ for the constraint
\eqref{eq:opt} and $\lambda(x)$ for the normalization of $p(w|x)$,
the explicit Lagrangian is
\begin{align}
    \mathcal{L}(p(w|x),\beta,\lambda(x))
    := I(X;W) + \beta\bigl(I(X;Y|W) - \varepsilon\bigr)
    + \sum_x \lambda(x)\left(\sum_w p(w|x) - 1\right).
    \label{eq:lagrangian}
\end{align}
Let $R(\varepsilon)$ denote the optimal value of \eqref{eq:opt} as a
function of~$\varepsilon$:
\begin{align}
    R(\varepsilon) := \inf_{p(w|x):\, I(X;Y|W)\leq\varepsilon} I(X;W).
    \label{eq:rd_function}
\end{align}
For finite discrete alphabets, the infimum is attained and ``inf"
may be replaced by ``min". We use ``min" throughout the main text
with this implicit assumption.
The Lagrange multiplier satisfies $\beta=-R'(\varepsilon)$, as shown in
Remark~\ref{rem:beta_slope} below.
As we shall show, $\beta$ parameterizes this curve monotonically, so
specifying $\beta$ is equivalent to specifying $\varepsilon$ at all
smooth points.
In practice, $\beta$ is therefore treated as a hyperparameter that
selects the desired compression-prediction tradeoff.

Setting the functional derivative of \eqref{eq:lagrangian} with
respect to $p(w|x)$ to zero yields the optimal encoder (Theorem 4 of
\cite{tishby1999information}): it assigns higher probability to latent
values $w$ whose predictive distribution $p(y|w)$ is close to $p(y|x)$
in KL divergence.
The explicit form is derived in the next section.

\begin{remark}[Lagrange multiplier as slope]
\label{rem:beta_slope}
From \eqref{eq:lagrangian}, $\partial\mathcal{L}/\partial\varepsilon=-\beta$.
At the optimal encoder the constraint is binding, so $\mathcal{L}=I(X;W)=R(\varepsilon)$.
At smooth points, differentiating with respect to $\varepsilon$ gives
$R'(\varepsilon)=-\beta$ \cite[Ch.~10]{cover2006elements}.
The multiplier is the negative slope of the rate-distortion function:
larger $\beta$ imposes a steeper rate penalty and drives $\varepsilon$ toward zero.
\end{remark}

%%%%%%%%%%%%%%%%%%%%%%%%%%%%%%%%%%%%%%%%%%%%%%%%%%%%%%%%%%%%

\section{IB as Rate Distortion}

\subsection{Rate-Distortion Equivalence}

The optimization \eqref{eq:opt} is equivalent to a classical
rate-distortion problem \cite{shannon1959coding,berger1971rate} with a
task-dependent distortion measure \cite{harremoes2007information}.
Define the distortion between input $x$ and latent value $w$ as
\begin{align}
    d(x,w) := D_{\mathrm{KL}}\bigl(p(y|x)\,\|\,p(y|w)\bigr).
    \label{eq:distortion}
\end{align}
This measures predictive mismatch: $d(x,w)=0$ if and only if
$p(y|w)=p(y|x)$.

\begin{proposition}[Rate-distortion equivalence]
\label{prop:rd}
Under the Markov chain \eqref{eq:markov},
\begin{align}
    \mathbb{E}[d(X,W)] = H(Y|W) - H(Y|X) = I(X;Y|W).
    \label{eq:rd_identity}
\end{align}
Therefore the optimization \eqref{eq:opt} is equivalent to
\begin{align}
    \min_{p(w|x)} \; I(X;W) \quad \text{s.t.} \quad
    \mathbb{E}[d(X,W)] \leq \varepsilon,
    \label{eq:rd}
\end{align}
and the Lagrangian \eqref{eq:lagrangian} takes the classical
rate-distortion form \cite{berger1971rate}
\begin{align}
    \mathcal{L} = I(X;W) + \beta\,\mathbb{E}[d(X,W)].
    \label{eq:rd_lagrangian}
\end{align}
\end{proposition}

\begin{proof}
Expanding the expectation and splitting the logarithm,
\begin{align}
    \mathbb{E}[d(X,W)]
    &= \sum_{x,w} p(x,w) \sum_y p(y|x)\log\frac{p(y|x)}{p(y|w)} \\
    &= \sum_{x,y} p(x)p(y|x)\log p(y|x)
    - \sum_{x,w,y} p(x,w)p(y|x)\log p(y|w).
\end{align}
The first term equals $-H(Y|X)$.
For the second, the Markov chain \eqref{eq:markov} gives
$p(x,w)p(y|x)=p(x,w,y)$, so summing over $x$,
\begin{align}
    \sum_{x,w,y} p(x,w,y)\log p(y|w)
    = \sum_{w,y} p(w,y)\log p(y|w) = -H(Y|W).
\end{align}
Hence $\mathbb{E}[d(X,W)] = H(Y|W) - H(Y|X) = I(X;Y|W)$.
\end{proof}

The distortion \eqref{eq:distortion} measures predictive mismatch
rather than reconstruction error: two inputs $x,x'$ can be merged
without distortion if and only if $p(y|x)=p(y|x')$.
By marginalization,
\begin{align}
    p(y|w) = \mathbb{E}[p(y|X)\,|\,W=w] = \sum_{x} p(x|w)\,p(y|x),
    \label{eq:decoder_exact}
\end{align}
so the decoder sees the average predictive distribution of inputs
mapped to $w$.

\subsection{The General Solution}
\label{sec:general_solution}

Setting the functional derivative of \eqref{eq:rd_lagrangian} with
respect to $p(w|x)$ to zero yields the optimal encoder
\begin{align}
    p(w|x) = \frac{p(w)}{Z(x,\beta)}
    \exp\!\left(-\beta\,D_{\mathrm{KL}}
    \bigl(p(y|x)\,\|\,p(y|w)\bigr)\right),
    \label{eq:opt_encoder2}
\end{align}
where $Z(x,\beta)=\sum_{w'}p(w')\exp(-\beta
D_{\mathrm{KL}}(p(y|x)\|p(y|w')))$ is a normalization function and
$p(y|w)$ is the self-consistent decoder \eqref{eq:decoder_exact}.
Each input $x$ is assigned probabilistically to latent values $w$,
with higher probability given to $w$ whose predictive distribution
$p(y|w)$ is close to $p(y|x)$ in KL divergence.

The encoder \eqref{eq:opt_encoder2}, the marginal $p(w)$, and the
decoder $p(y|w)$ satisfy the self-consistency equations
\begin{align}
    p(w) &= \sum_x p(x)\,p(w|x),
    \label{eq:marginal}\\
    p(y|w) &= \frac{1}{p(w)}\sum_x p(x)\,p(w|x)\,p(y|x).
    \label{eq:decoder}
\end{align}
These can be solved iteratively via the Blahut-Arimoto algorithm
\cite{blahut1972computation,arimoto1972algorithm}: starting from an
initial $p(w)$ and $p(y|w)$, the algorithm alternates the encoder
update \eqref{eq:opt_encoder2}, the marginal update \eqref{eq:marginal},
and the decoder update \eqref{eq:decoder}.
Convergence to the optimum of \eqref{eq:rd_lagrangian} for fixed $\beta$
is guaranteed.

\begin{remark}
\label{rem:stochastic}
For $\varepsilon\in(0,I(X;Y))$, the optimal encoder achieving
the rate-distortion function \eqref{eq:rd_function} is generally stochastic
\cite{strouse2017deterministic}.
A deterministic encoder assigns each $x$ to a fixed point in
$\mathcal{P}_K$, restricting $W$ to a finite subset of the predictive
manifold $\mathcal{M}$ and typically preventing it from tracing the full
rate-distortion curve.
At the endpoint $\varepsilon=0$, the optimal encoder is deterministic,
as we shall show in Theorem~\ref{thm:mss} in the next section.
\end{remark}

\subsection{Minimal Sufficient Statistic: The Zero-Distortion Case}

At $\varepsilon=0$, the distortion constraint becomes exact sufficiency, i.e.
$I(X;Y|W)=0$.

\begin{theorem}[Minimal sufficient statistic]
\label{thm:mss}
Define $f^*:\mathcal{X}\to\mathcal{P}_K$ by $f^*(x):=p(Y|x)$, and let
$\{C_i\}$ denote the partition of $\mathcal{X}$ into level sets of $f^*$:
\begin{align}
    x \sim x' \quad \Longleftrightarrow \quad f^*(x)=f^*(x'),
    \label{eq:equiv}
\end{align}
where $C_i=\{x:f^*(x)=q_i\}$ for some $q_i\in\mathcal{P}_K$.
Define the random variable
\begin{align}
    W^* := f^*(X),
    \label{eq:mss}
\end{align}
so that $W^*$ takes values in $\mathcal{P}_K$. Then:
\begin{enumerate}
    \item $W^*$ is a sufficient statistic of $X$ for $Y$:
    $I(X;Y|W^*)=0$.
    \item $W^*$ is minimal: any encoder $W$ with $I(X;Y|W)=0$
    satisfies $W^*=g(W)$ almost surely, where $g(w):=p(Y|w)$ is the
    self-consistent decoder \eqref{eq:decoder}.
\end{enumerate}
\end{theorem}

\begin{proof}
\textit{(1) Sufficiency.}
Since $W^*=p(Y|X)$, conditioning on $W^*=t$ fixes $p(Y|x)=t$ for
every $x$ in the preimage $\{x:p(Y|x)=t\}$, so $p(Y|W^*=t)=t$.
Hence $H(Y|W^*)=H(Y|X)$, giving $I(X;Y|W^*)=0$.

\textit{(2) Minimality.}
Suppose $W$ satisfies $I(X;Y|W)=0$, equivalently
$\mathbb{E}[d(X,W)]=0$ by Proposition~\ref{prop:rd}.
Since $d(x,w)\geq 0$ and its expectation is zero, $d(x,w)=0$ almost
surely, so $p(y|w)=p(y|x)$ for every $(x,w)$ with $p(x,w)>0$.
The decoder $g(w)=p(Y|w)$ therefore satisfies $g(w)=p(Y|x)$
for every $x$ mapped to $w$, giving $W^*=g(W)$ almost surely.
\end{proof}

\begin{corollary}[Minimum rate]
\label{cor:min_rate}
$W^*$ achieves the minimum rate among all sufficient encoders:
$I(X;W^*)=H(W^*)$, and any $W$ with $I(X;Y|W)=0$ satisfies
$I(X;W)\geq H(W^*)$, with equality if and only if $I(X;W|W^*)=0$.
\end{corollary}

\begin{proof}
By Theorem~\ref{thm:mss}, $W^*=g(W)$, so the data processing
inequality applied to $X\to W\to W^*$ gives
$I(X;W)\geq I(X;W^*)$.
Since $W^*=f^*(X)$ is deterministic, $H(W^*|X)=0$, hence
$I(X;W^*)=H(W^*)$.
Equality $I(X;W)=H(W^*)$ holds if and only if $I(X;W|W^*)=0$,
by the chain rule $I(X;W)=H(W^*)+I(X;W|W^*)$.
\end{proof}

\begin{corollary}[Hard assignment in the large-$\beta$ limit]
\label{cor:hard}
Let $W$ correspond to the partition $\{C_i\}$ of Theorem~\ref{thm:mss},
with each cell $C_i$ associated with a distinct latent value $w_i$
satisfying $p(y|w_i)=q_i(y)$, and let $x_i\in C_i$ be any
representative.
Then
\begin{align}
    D_{\mathrm{KL}}\bigl(p(y|x_i)\,\|\,p(y|w_j)\bigr) =
    \begin{cases} 0 & i=j \\ >0 & i\neq j, \end{cases}
    \label{eq:kl_cases}
\end{align}
and the optimal encoder \eqref{eq:opt_encoder2} satisfies
\begin{align}
    \lim_{\beta\to\infty} p(w_j|x_i) =
    \begin{cases} 1 & i=j \\ 0 & i\neq j, \end{cases}
\end{align}
i.e.\ the encoder performs hard assignment to the correct partition cell.
\end{corollary}

\begin{proof}
Since $p(y|w_i)=q_i(y)=p(y|x_i)$ by construction,
$D_{\mathrm{KL}}(p(y|x_i)\|p(y|w_i))=0$.
For $i\neq j$, $p(y|w_j)=q_j(y)\neq q_i(y)=p(y|x_i)$ by definition
of the partition, so $D_{\mathrm{KL}}(p(y|x_i)\|p(y|w_j))>0$.
As $\beta\to\infty$, $\exp(-\beta D_{\mathrm{KL}})\to 0$ for all
$j\neq i$ in \eqref{eq:opt_encoder2}, so the distribution concentrates
on $w_i$.
\end{proof}

\begin{corollary}[Dimensionality bound]
\label{thm:dim}
The intrinsic dimension of the optimal representation $W^*$ is at
most $K-1$.
\end{corollary}

\begin{proof}
By \eqref{eq:mss}, $W^*=p(Y|X)$ takes values in $\mathcal{P}_K$, which has dimension $K-1$.
For discrete $Y$ with $|\mathcal{Y}|=K$, each $p(Y|x)$ is a
$K$-dimensional probability vector, so the bound is immediate.
For continuous $Y$, $W^*=p(Y|X)$ can always be losslessly
reparameterized as $W=X$, so the intrinsic dimension of the optimal
representation is at most $d_X$.
When $f^*$ is Lipschitz, Appendix~\ref{app:eff_dim} gives the sharper
covering-number statement $d_{\mathrm{eff}}(\mathcal{M})\leq d_X$ via
Proposition~\ref{prop:geom_dim}, and hence
$K-1=\lceil d_{\mathrm{eff}}(\mathcal{M})\rceil\leq\lceil d_X\rceil$.
\end{proof}

\begin{remark}
The bound $\dim(W^*)\leq K-1$ is a statement about coordinate
dimension, distinct from the entropy $H(W^*)$, which measures the
diversity of predictive distributions and can be large even when $K$
is small.
In the infinite-data limit, any encoder-decoder network trained for
prediction will collapse onto a manifold of dimension at most $K-1$,
regardless of the ambient dimension of the latent space.
\end{remark}

\subsection{Geometric Interpretation}
\label{sec:geometric}

The optimal encoder at $\varepsilon=0$ computes $W^*=p(Y|X)$, mapping
inputs onto the predictive manifold
$\mathcal{M}\subset\mathcal{P}_K$.
All information in $X$ relevant for predicting $Y$ lives on $\mathcal{M}$.
Nuisance variability is discarded.

For any $\varepsilon\geq 0$, the optimal encoder \eqref{eq:opt_encoder2}
depends on $x$ only through $p(Y|x)=W^*\in\mathcal{P}_K$.
The encoder therefore factors as
\begin{align}
    X \xrightarrow{f^*} W^* \in \mathcal{P}_K
    \xrightarrow{\text{stochastic kernel } p(w|w^*)} W.
\end{align}
The stochastic kernel $p(w|w^*)$ defines a soft clustering of $\mathcal{M}$.
The decoder
\begin{align}
    p(y|w) = \mathbb{E}[W^*\,|\,W=w] \in \mathcal{P}_K
    \label{eq:decoder_center}
\end{align}
is the center of the cluster associated with $w$.
The decoder law $p(y|w)$ is the only property of $w$ relevant for
predicting $Y$, so $W$ takes values in $\mathcal{P}_K$ without loss
of generality.
Any two values $w\neq w'$ with $p(y|w)=p(y|w')$ carry identical
prediction. Merging them leaves the rate and distortion unchanged.

The encoder-decoder system therefore takes the canonical form
\begin{align}
    X \xrightarrow{\text{nonlinear encoder}} W\in\mathcal{P}_K
    \xrightarrow{\text{linear decoder}} \hat{Y}
    \label{eq:canonical}
\end{align}
for all $\varepsilon\geq 0$, concentrating all representational
complexity in the encoder.

At $\varepsilon=0$, the stochastic kernel collapses to a point mass:
$W=W^*$ and the encoder computes $p(Y|X)$ exactly.
As $\varepsilon$ increases, the clusters broaden and overlap, merging
nearby predictive distributions at the cost of distortion.
At $\varepsilon=I(X;Y)$, all inputs merge into a single cluster,
giving $W\perp X$ (see Figure~\ref{fig:soft_clustering}).
In the canonical gauge $W^*=p(Y|X)$,
the optimal decoder reduces to indexing into the latent vector,
$p(y|w)=e_y^\top w$. For $\varepsilon>0$ the latent need not equal
the predictive distribution. The decoder is then the composition of an
encoder-determined map $w\mapsto\mathbb{E}[W^*\mid W=w]$ with the same
linear readout.

\begin{remark}[Gauge freedom]
\label{rem:gauge_decoder}
Decoder linearity is a property of the canonical representation
$W=p(Y|X)\in\mathcal{P}_K$.
The IB objective is invariant under any invertible reparameterization
$r\colon W\mapsto r(W)$: if $W$ is sufficient for $Y$, so is $r(W)$.
A loss depending only on mutual information terms does not pin the
gauge. The encoder may converge to any informationally equivalent
representation, and the decoder is then the generally nonlinear map
$e_y^\top\circ r^{-1}$.
The geometric picture of $\mathcal{M}\subset\mathcal{P}_K$ therefore
describes an \emph{equivalence class}: any optimal representation is
related to the canonical form by an invertible map.
\end{remark}

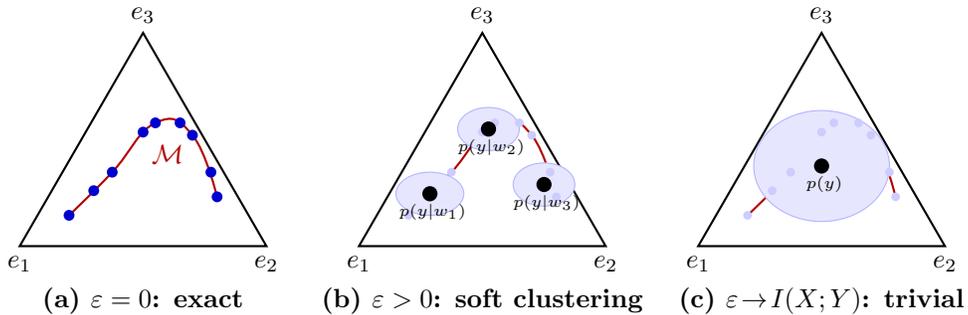
\begin{figure}[h]
\centering
\begin{tikzpicture}[scale=0.82]
% Panel (a): epsilon = 0
\begin{scope}[shift={(0,0)}]
    \draw[thick] (0,0) -- (4,0) -- (2,3.46) -- cycle;
    \node[below] at (0,0) {\small$e_1$};
    \node[below] at (4,0) {\small$e_2$};
    \node[above] at (2,3.46) {\small$e_3$};

    \draw[thick, red!70!black] plot[smooth, tension=0.8]
        coordinates {(0.8,0.5) (1.5,1.2) (2.2,2.0) (2.8,1.8) (3.2,0.8)};
    \node[red!70!black, font=\small] at (2.4,1.5) {$\mathcal{M}$};

    \foreach \x/\y in {0.8/0.5, 1.2/0.9, 1.5/1.2, 2.0/1.85,
                        2.2/2.0, 2.6/2.0, 2.8/1.8, 3.1/1.2, 3.2/0.8}
        \fill[blue!80!black] (\x,\y) circle (2.5pt);

    \node[below, font=\small\bfseries] at (2,-0.5)
        {(a) $\varepsilon=0$: exact};
\end{scope}

% Panel (b): small epsilon
\begin{scope}[shift={(5.5,0)}]
    \draw[thick] (0,0) -- (4,0) -- (2,3.46) -- cycle;
    \node[below] at (0,0) {\small$e_1$};
    \node[below] at (4,0) {\small$e_2$};
    \node[above] at (2,3.46) {\small$e_3$};

    \draw[thick, red!70!black] plot[smooth, tension=0.8]
        coordinates {(0.8,0.5) (1.5,1.2) (2.2,2.0) (2.8,1.8) (3.2,0.8)};

    \draw[blue!30, fill=blue!10] (1.15,0.85) ellipse (0.55 and 0.35);
    \draw[blue!30, fill=blue!10] (2.1,1.9) ellipse (0.5 and 0.35);
    \draw[blue!30, fill=blue!10] (3.0,1.0) ellipse (0.5 and 0.35);

    \foreach \x/\y in {0.8/0.5, 1.2/0.9, 1.5/1.2, 2.0/1.85,
                        2.2/2.0, 2.6/2.0, 2.8/1.8, 3.1/1.2, 3.2/0.8}
        \fill[blue!20] (\x,\y) circle (2pt);

    \fill[black] (1.15,0.85) circle (3.5pt)
        node[below, font=\tiny] {$\;p(y|w_1)$};
    \fill[black] (2.1,1.9) circle (3.5pt)
        node[below, font=\tiny] {$\;p(y|w_2)$};
    \fill[black] (3.0,1.0) circle (3.5pt)
        node[below, font=\tiny] {$\;p(y|w_3)$};

    \node[below, font=\small\bfseries] at (2,-0.5)
        {(b) $\varepsilon>0$: soft clustering};
\end{scope}

% Panel (c): large epsilon
\begin{scope}[shift={(11,0)}]
    \draw[thick] (0,0) -- (4,0) -- (2,3.46) -- cycle;
    \node[below] at (0,0) {\small$e_1$};
    \node[below] at (4,0) {\small$e_2$};
    \node[above] at (2,3.46) {\small$e_3$};

    \draw[thick, red!70!black] plot[smooth, tension=0.8]
        coordinates {(0.8,0.5) (1.5,1.2) (2.2,2.0) (2.8,1.8) (3.2,0.8)};

    \draw[blue!30, fill=blue!10] (2.0,1.3) ellipse (1.1 and 0.9);

    \foreach \x/\y in {0.8/0.5, 1.2/0.9, 1.5/1.2, 2.0/1.85,
                        2.2/2.0, 2.6/2.0, 2.8/1.8, 3.1/1.2, 3.2/0.8}
        \fill[blue!20] (\x,\y) circle (2pt);

    \fill[black] (2.0,1.3) circle (3.5pt)
        node[below, font=\tiny] {$\;p(y)$};

    \node[below, font=\small\bfseries] at (2,-0.5)
        {(c) $\varepsilon\!\to\!I(X;Y)$: trivial};
\end{scope}
\end{tikzpicture}
\caption{Soft clustering of the predictive manifold
$\mathcal{M}\subset\mathcal{P}_3$ at three distortion levels.
Each triangle is the probability simplex $\mathcal{P}_3$ with
vertices $e_1,e_2,e_3$.
The red curve is $\mathcal{M}=\{p(Y|x):x\in\mathcal{X}\}$.
Blue dots are individual predictive distributions $p(Y|x_i)$.
\textbf{(a)}~At $\varepsilon=0$, the encoder computes $W^*=p(Y|X)$
exactly (see Theorem~\ref{thm:mss}).
\textbf{(b)}~At moderate $\varepsilon$, nearby points merge into
soft clusters. Black dots are cluster centers
$p(y|w)=\mathbb{E}[W^*|W{=}w]$ \eqref{eq:decoder_center}.
\textbf{(c)}~As $\varepsilon\to I(X;Y)$, all points merge into a
single cluster centered at the marginal $p(y)$, and $W\perp X$.}
\label{fig:soft_clustering}
\end{figure}

\subsection{Invariances of the Optimal Representation}

The optimal encoder discards all variations in $X$ that leave
$p(Y|x)$ unchanged, at every distortion level.

\begin{proposition}[Invariance]
\label{prop:invariance}
Let $T:\mathcal{X}\to\mathcal{X}$ satisfy
\begin{align}
    p(Y|T(x)) = p(Y|x) \quad \forall\, x.
    \label{eq:invariance_cond}
\end{align}
Then $p(w|T(x))=p(w|x)$ for all $w$ and all $\varepsilon\geq 0$.
In particular, $W^*(T(x))=W^*(x)$ at $\varepsilon=0$.
\end{proposition}

\begin{proof}
The optimal encoder depends on $x$ only through $p(y|x)$: at
$\varepsilon=0$ via $W^*=p(Y|X)$, and for $\varepsilon>0$ via the
self-consistency equation \eqref{eq:opt_encoder2}.
Since $p(Y|T(x))=p(Y|x)$, the claim follows in both cases.
\end{proof}

\begin{remark}
Transformations satisfying \eqref{eq:invariance_cond} are nuisance
variations: changes in $X$ that carry no information about $Y$.
When they form a group $G$ acting on $\mathcal{X}$, the canonical
latent space $\mathcal{M}\subset\mathcal{P}_K$ is the image of the
quotient $\mathcal{X}/G$ under $f^*$, and the encoder collapses
$G$-orbits at every distortion level.
This is the input-side counterpart of gauge freedom
(see Remark~\ref{rem:gauge_decoder}): gauge freedom concerns invertible
reparameterizations of the output $W$. The present invariance concerns
the collapsing of input $X$.
\end{remark}

\begin{example}[Ternary Classification on the Simplex]
\label{ex:ternary_example}

Let $X$ be uniform on $\{1,2,3,4,5,6\}$ and $Y\in\{1,2,3\}$ with
\begin{align}
    p(Y|X{=}1) = p(Y|X{=}2) &= (0.8,\;0.1,\;0.1), \notag\\
    p(Y|X{=}3) = p(Y|X{=}4) &= (0.1,\;0.8,\;0.1), \notag\\
    p(Y|X{=}5) = p(Y|X{=}6) &= (0.1,\;0.1,\;0.8).
    \label{eq:ternary_ex}
\end{align}
The minimal sufficient statistic $W^*=p(Y|X)$ takes three values
$w_1=(0.8,0.1,0.1)$, $w_2=(0.1,0.8,0.1)$, $w_3=(0.1,0.1,0.8)$,
each with probability $1/3$.
The predictive manifold $\mathcal{M}=\{w_1,w_2,w_3\}$ is this finite
set of three points, each lying near a vertex of $\mathcal{P}_3$
(see Figure~\ref{fig:ternary}).
The marginal $p(Y)=(1/3,1/3,1/3)$ is the centroid of that triangle.
At $\varepsilon=0$, the encoder maps each $x$ to its cluster
center $w_i$.
As $\varepsilon$ increases, the clusters broaden and the representation
coarsens toward the single point $p(Y)$.

\begin{figure}[h]
\centering
\begin{tikzpicture}[scale=1.1]
    \coordinate (e1) at (0,0);
    \coordinate (e2) at (4,0);
    \coordinate (e3) at (2,3.46);
    \draw[thick] (e1) -- (e2) -- (e3) -- cycle;
    \node[below left] at (e1) {$e_1=(1,0,0)$};
    \node[below right] at (e2) {$e_2=(0,1,0)$};
    \node[above] at (e3) {$e_3=(0,0,1)$};

    % w = (p1,p2,p3) -> p1*e1 + p2*e2 + p3*e3
    % w1=(0.8,0.1,0.1): 0.8(0,0)+0.1(4,0)+0.1(2,3.46) = (0.6,0.346)
    \coordinate (w1) at (0.6, 0.346);
    % w2=(0.1,0.8,0.1): 0.1(0,0)+0.8(4,0)+0.1(2,3.46) = (3.4,0.346)
    \coordinate (w2) at (3.4, 0.346);
    % w3=(0.1,0.1,0.8): 0.1(0,0)+0.1(4,0)+0.8(2,3.46) = (2.0,2.768)
    \coordinate (w3) at (2.0, 2.768);

    \draw[thick, red!70!black, dashed] (w1) -- (w2) -- (w3) -- cycle;

    \fill[blue!80!black] (w1) circle (3pt)
        node[below left, font=\small] {$w_1$};
    \fill[blue!80!black] (w2) circle (3pt)
        node[below right, font=\small] {$w_2$};
    \fill[blue!80!black] (w3) circle (3pt)
        node[above, font=\small] {$w_3$};

    % Marginal p(Y) = (1/3,1/3,1/3) -> center (2, 1.153)
    \fill[black!60] (2, 1.153) circle (2.5pt)
        node[below, font=\small] {$p(Y)$};

    % M is the three blue points; dashed triangle is their convex hull (visual aid)
\end{tikzpicture}
\caption{Example \eqref{eq:ternary_ex} inside $\mathcal{P}_3$.
Blue points: $\mathcal{M}=\{w_1,w_2,w_3\}$, the predictive manifold.
Dashed triangle: convex hull of $\mathcal{M}$ (visual aid).
Grey dot: marginal $p(Y)=(1/3,1/3,1/3)$.}
\label{fig:ternary}
\end{figure}
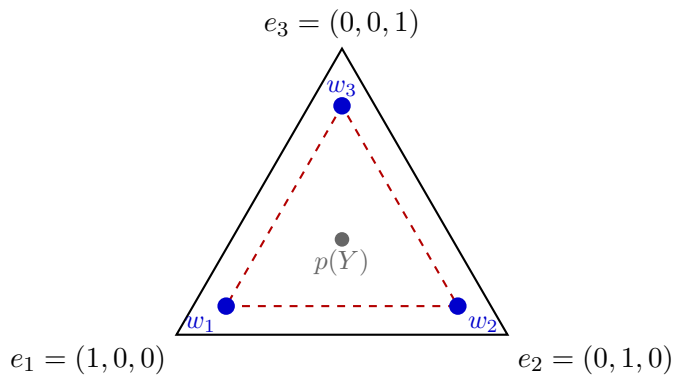

\end{example}

%%%%%%%%%%%%%%%%%%%%%%%%%%%%%%%%%%%%%%%%%%%%%%%%%%%%%%%%%%%%

\section{The Information Plane}

\subsection{The Achievable Region}
\label{sec:achievable}

The \textit{achievable region} is the set of all pairs $(I(X;W),I(W;Y))$
attainable by some encoder $p(w|x)$ satisfying the Markov chain
\eqref{eq:markov}, for a fixed joint distribution $p(x,y)$.

\medskip
\noindent\textbf{Upper bound (converse).}
The data processing inequality \cite{cover2006elements} applied to
the two ends of $Y-X-W$ gives two bounds:
\begin{align}
    I(W;Y) &\leq I(X;Y), \label{eq:dpi_horizontal}\\
    I(W;Y) &\leq I(X;W). \label{eq:dpi_diagonal}
\end{align}
No encoder extracts more information about $Y$ than $X$ contains.
Nor can it convey more about $Y$ than about $X$.
These bounds confine the achievable region to
\begin{align}
    \mathcal{C} := \bigl\{(R,\Delta) : 0\leq \Delta \leq \min(R,\;I(X;Y))\bigr\},
    \label{eq:converse}
\end{align}
where $R:=I(X;W)$ and $\Delta:=I(W;Y)$.
The boundary of $\mathcal{C}$ consists of the diagonal
$\Delta=R$ for $R\leq I(X;Y)$ and the horizontal line
$\Delta=I(X;Y)$ for $R\geq I(X;Y)$.

\medskip
\noindent\textbf{Achievable endpoints.}
Two operating points are always achievable:
\begin{itemize}
    \item \textit{Trivial encoder} ($W\perp X$): gives $(R,\Delta)=(0,0)$.
    \item \textit{Minimal sufficient statistic} ($W=W^*=p(Y|X)$): gives $(R,\Delta)=(H(W^*),\;I(X;Y))$,
    by Theorem~\ref{thm:mss}, i.e. $I(X;W^*)=H(W^*)$ and $I(W^*;Y)=I(X;Y)$.
\end{itemize}

\medskip
\noindent\textbf{The gap.}
The endpoint $(H(W^*),I(X;Y))$ lies on the diagonal \eqref{eq:dpi_diagonal}
if and only if $H(W^*)=I(X;Y)$, equivalently $H(W^*|Y)=0$.
In general,
\begin{align}
    H(W^*) = I(X;Y) + H(W^*|Y),
    \label{eq:gap_identity}
\end{align}
since $H(W^*)=I(W^*;Y)+H(W^*|Y)=I(X;Y)+H(W^*|Y)$.
The gap $H(W^*|Y)\geq 0$ is the entropy overhead that the optimal encoder spends on predictive distributions that
are indistinguishable from their outputs alone.
When $H(W^*|Y)>0$ the endpoint lies strictly below the diagonal.

\medskip
\noindent\textbf{Time sharing lower bound.}
Mixing $W^*$ with the trivial encoder with probability $\lambda\in[0,1]$,
independently of $(X,Y)$, achieves every point
\begin{align}
    \bigl(\lambda H(W^*),\;\lambda I(X;Y)\bigr),
    \quad\lambda\in[0,1].
    \label{eq:timesharing}
\end{align}
This traces a line from $(0,0)$ to $(H(W^*),I(X;Y))$ with slope
\begin{align}
    \frac{I(X;Y)}{H(W^*)} = \frac{I(X;Y)}{I(X;Y)+H(W^*|Y)} \leq 1,
    \label{eq:ts_slope}
\end{align}
with equality if and only if $H(W^*|Y)=0$.
When the slope is strictly below~$1$, the time-sharing line lies below
the diagonal. A \textit{triangular gap} opens between the converse
boundary $\mathcal{C}$ and the time-sharing line.

\medskip
\noindent\textbf{The IB curve \textnormal{(see Section~\ref{sec:ibcurve})}.}
The time-sharing line \eqref{eq:timesharing} is a lower bound on the
IB curve. The converse boundary $\mathcal{C}$ is an upper bound.
The gap between them depends on $H(W^*|Y)$.
By Theorem~\ref{thm:convex}, the IB curve is concave and lies within
this gap. Its exact shape depends on $p(x,y)$.

\medskip
Figure~\ref{fig:ib_curve} summarises the information plane structure: the converse
boundary $\mathcal{C}$, the time-sharing lower bound, the IB curve
(concave, connecting $(0,0)$ to $(H(W^*),I(X;Y))$), and the
triangular gap that opens when $H(W^*|Y)>0$.

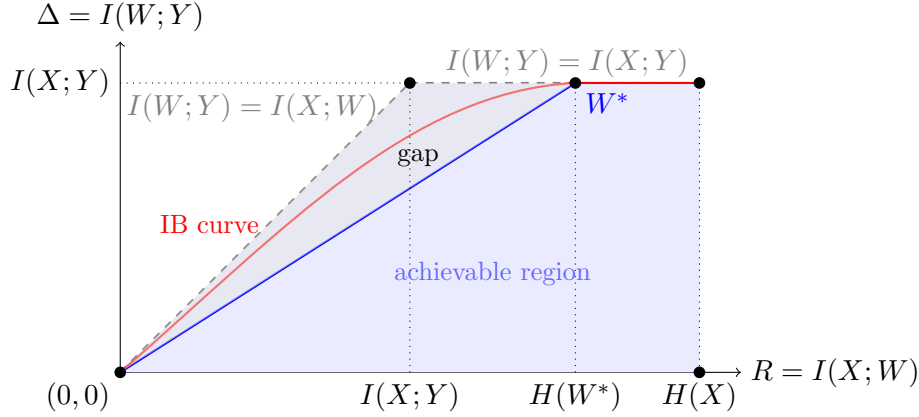
\begin{figure}[h]
\centering
\begin{tikzpicture}[scale=1.1]
% axes
\draw[->] (0,0) -- (7.5,0) node[right] {$R=I(X;W)$};
\draw[->] (0,0) -- (0,4.0) node[above] {$\Delta=I(W;Y)$};
% achievable region (trapezoid fill)
\fill[blue!8] (0,0) -- (3.5,3.5) -- (7,3.5) -- (7,0) -- cycle;
% converse boundary
\draw[thick,gray,dashed] (0,0) -- (3.5,3.5);
\draw[thick,gray,dashed] (3.5,3.5) -- (7,3.5);
\node[gray,above] at (1.6,2.9) {$I(W;Y)=I(X;W)$};
\node[gray,above] at (5.4,3.45) {$I(W;Y)=I(X;Y)$};
% time sharing line
\draw[thick,blue] (0,0) -- (5.5,3.5);
% IB curve (concave, above time-sharing, below converse)
\draw[thick,red] (0,0) .. controls (2.0,1.75) and (3.5,3.4) .. (5.5,3.5);
% flat portion of IB curve
\draw[thick,red] (5.5,3.5) -- (7,3.5);
\node[red,left] at (1.8,1.8) {\small IB curve};
% gap shading (between converse and time-sharing)
\fill[gray!20,opacity=0.5] (0,0) -- (3.5,3.5) -- (5.5,3.5) -- cycle;
% key points
\fill (0,0) circle (2pt) node[below left] {$(0,0)$};
\fill (5.5,3.5) circle (2pt);
\fill (3.5,3.5) circle (2pt);
\fill (7,3.5) circle (2pt);
\fill (7,0) circle (2pt);
% labels
\node[left] at (0,3.5) {$I(X;Y)$};
\draw[dotted] (0,3.5) -- (3.5,3.5);
\node[below] at (5.5,0) {$H(W^*)$};
\draw[dotted] (5.5,0) -- (5.5,3.5);
\node[below] at (3.5,0) {$I(X;Y)$};
\draw[dotted] (3.5,0) -- (3.5,3.5);
\node[below] at (7,0) {$H(X)$};
\draw[dotted] (7,0) -- (7,3.5);
% annotations
\node[blue,below right] at (5.5,3.5) {$W^*$};
\node at (3.6,2.6) {\small gap};
\node[blue!60] at (4.5,1.2) {\small achievable region};
\end{tikzpicture}
\caption{The achievable region (light blue trapezoid) in the
$(I(X;W),I(W;Y))$ plane, bounded by the converse boundary
$\mathcal{C}$ (gray dashed: diagonal and horizontal data processing
inequality bounds),
the $R$-axis, and $R=H(X)$.
The blue line is the time-sharing lower bound on the IB curve.
The darker shaded triangle is the gap region between the converse
boundary and the time-sharing line.
The true IB curve (red) is concave, connects $(0,0)$ to
$(H(W^*),I(X;Y))$, and is flat at $I(X;Y)$ for $R\in[H(W^*),H(X)]$.
When $H(W^*|Y)=0$, the gap triangle collapses and the IB curve
coincides with the converse boundary up to $R=I(X;Y)=H(W^*)$.}
\label{fig:ib_curve}
\end{figure}

\begin{example}[Binary Classification and the Information Plane]
\label{ex:binary_example}

Let $X$ be uniform on $\{1,2,3,4\}$ and $Y\in\{0,1\}$ with
\begin{align}
    p(Y|X{=}1) = p(Y|X{=}2) &= (0.1,\;0.9), \notag\\
    p(Y|X{=}3) = p(Y|X{=}4) &= (0.9,\;0.1).
    \label{eq:binary_ex}
\end{align}
The minimal sufficient statistic $W^*=p(Y|X)$ takes two values
$w_1=(0.1,0.9)$ and $w_2=(0.9,0.1)$, each with probability $1/2$, so
$H(W^*)=\log 2\approx 0.693$ nats.
The marginal satisfies $p(Y{=}1)=0.5$, so $H(Y)=\log 2$ and
$H(Y|X)=H_b(0.9)\approx 0.325$ nats, where
$H_b(p)=-p\log p-(1-p)\log(1-p)$, so $I(X;Y)\approx 0.368$ nats.

The endpoint $(H(W^*),I(X;Y))\approx(0.693,\;0.368)$ lies strictly
below the diagonal.
The gap is $H(W^*|Y)=H(W^*)-I(X;Y)\approx 0.325$ nats:
given $Y=1$, the posterior $p(W^*{=}w_1|Y{=}1)=0.9\neq 1$,
so observing $Y$ does not determine which predictive distribution
generated it.

The time-sharing line has slope $I(X;Y)/H(W^*)\approx 0.531$, and the
triangular gap has vertices $(0,0)$, $(0.368,\;0.368)$, and $(0.693,\;0.368)$.
The IB curve lies in this gap.

\end{example}

\subsection{The IB Curve}
\label{sec:ibcurve}

Define the \textit{IB curve} as the upper boundary of the achievable
region: for each rate $R\geq 0$,
\begin{align}
    \Delta^*(R) := \max_{p(w|x):\,I(X;W)\leq R} I(W;Y).
    \label{eq:ib_curve_def}
\end{align}
Equivalently, in terms of rate-distortion,
\begin{align}
    \Delta^*(R) = I(X;Y) - \varepsilon^*(R),
    \label{eq:ib_rd}
\end{align}
where $\varepsilon^*(R)$ is the minimum distortion achievable at rate
$R$ (inverse of \eqref{eq:rd_function}).

\begin{theorem}[Concavity of the IB curve]
\label{thm:convex}
The IB curve $\Delta^*(R)$ is concave and non-decreasing in $R$.
Equivalently, the rate-distortion function $R(\varepsilon)$ \eqref{eq:rd_function} is convex and
non-increasing in $\varepsilon\in[0,I(X;Y)]$.
\end{theorem}

\begin{proof}
    See Appendix~\ref{app:convexity}.
\end{proof}

\medskip
\noindent\textbf{Bounds of the IB curve.}
Combined with the bounds of Section~\ref{sec:achievable},
Theorem~\ref{thm:convex} gives the following characterization:
\begin{itemize}
    \item The IB curve starts at $(0,0)$ and ends at
    $(H(W^*),I(X;Y))$.
    \item It is concave and lies weakly above the achievable time-sharing line
    \eqref{eq:timesharing}.
    \item It lies weakly below the converse boundary $\mathcal{C}$
    \eqref{eq:converse}.
    \item For $R\geq H(W^*)$, the curve is flat at
    $\Delta^*(R)=I(X;Y)$: one can augment $W^*$ with additional
    information about $X$ to increase the rate without affecting
    the predictive information (see Appendix~\ref{app:flat_portion}).
    \item The curve lies on the diagonal if and only if $H(W^*|Y)=0$. The gap then vanishes and $H(W^*)=I(X;Y)$.
\end{itemize}

The exact shape of the IB curve within the gap region depends on
$p(x,y)$ through the geometry of the predictive manifold
$\mathcal{M}\subset\mathcal{P}_K$ and has no closed form in general.
It can be computed iteratively via the Blahut-Arimoto algorithm \cite{blahut1972computation,arimoto1972algorithm}.

\medskip
\noindent\textbf{Slope of the IB curve.}
The normalization term in \eqref{eq:lagrangian} vanishes at any feasible $p(w|x)$.
Using $I(X;Y|W)=I(X;Y)-I(W;Y)$, the remaining terms give
\begin{align}
    \mathcal{L} = R - \beta\,\Delta + \beta\bigl(I(X;Y)-\varepsilon\bigr).
    \label{eq:lagrangian_rearranged}
\end{align}
Assuming the minimum in \eqref{eq:rd_function} is attained, each
$\varepsilon\in[0,I(X;Y)]$ determines a unique Lagrange multiplier
$\beta(\varepsilon):=-R'(\varepsilon)\geq 0$ (see Remark~\ref{rem:beta_slope}).
The level sets of $R-\beta\Delta$ in \eqref{eq:lagrangian_rearranged} are lines of slope
$1/\beta$ in the $(R,\Delta)$ plane.
The optimal encoder for distortion $\varepsilon$ is the point where such a line is tangent to the IB curve.
Denoting this tangent line $\Delta_\varepsilon(R)$,
\begin{align}
    \Delta_\varepsilon(R) = \frac{R}{\beta(\varepsilon)} + c(\varepsilon),
    \qquad
    c(\varepsilon) := I(X;Y) - \varepsilon - \frac{R(\varepsilon)}{\beta(\varepsilon)},
    \label{eq:ib_tangent}
\end{align}
where $c(\varepsilon)$ is the $\Delta$-axis intercept (see Figure~\ref{fig:ib_tangent}).

\begin{figure}[h]
\centering
\begin{tikzpicture}[scale=1.1]
% axes  (same as fig:ib_curve)
\draw[->] (0,0) -- (7.5,0) node[right] {$R=I(X;W)$};
\draw[->] (0,0) -- (0,4.0) node[above] {$\Delta=I(W;Y)$};
% IB curve  (same Bezier as fig:ib_curve)
\draw[thick,red] (0,0) .. controls (2.0,1.75) and (3.5,3.4) .. (5.5,3.5);
\draw[thick,red] (5.5,3.5) -- (7,3.5);
\node[red,left] at (1.8,1.8) {\small IB curve};
% endpoint reference lines  (same as fig:ib_curve)
\node[left] at (0,3.5) {$I(X;Y)$};
\draw[dotted] (0,3.5) -- (5.5,3.5);
\node[below] at (5.5,0) {$H(W^*)$};
\draw[dotted] (5.5,0) -- (5.5,3.5);
% Operating point at t≈0.5 on the Bezier: approximately (2.8, 2.4)
\coordinate (mid) at (2.8,2.4);
\fill[red] (mid) circle (2pt);
% \node[above right, font=\small] at (mid) {$(R_0,\,\Delta^*(R_0))$};
% Tangent line at mid: slope≈0.74, intercept c≈0.33
% y = 0.74*x + 0.33; endpoints: (0, 0.33) to (4.9, 3.95)
\draw[thick,blue] (0,0.33) -- (4.9,3.96);
\node[blue, right, font=\small] at (2.0,3.0) {slope $\tfrac{1}{\beta}$};
% y-intercept dot and label
\fill[blue] (0,0.33) circle (2pt);
\node[blue, left, font=\small] at (0,0.33) {$c(R_0)$};
% dotted drop lines from operating point to axes
\draw[dotted] (mid) -- (2.8,0) node[below, font=\small] {$R_0$};
\draw[dotted] (mid) -- (0,2.4) node[left, font=\small] {$\Delta^*(R_0)$};
\end{tikzpicture}
\caption{Tangent to the IB curve at operating point $(R_0,\Delta^*(R_0))$.
The slope of the tangent is $1/\beta$ and its $\Delta$-axis
intercept is $c(R_0)=\Delta^*(R_0)-R_0/\beta$.
As $R_0$ increases, $\beta$ grows, the slope flattens, and $c(R_0)$
rises from $0$ to $I(X;Y)$.}
\label{fig:ib_tangent}
\end{figure}
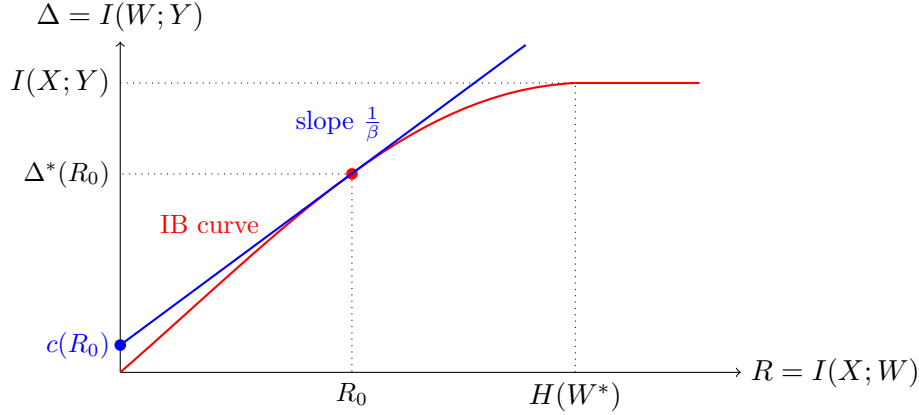

\medskip
\noindent\textbf{The critical $\beta$.}
Monotonicity alone does not imply that every $\beta>0$ yields a
non-trivial operating point.
Since the curve is concave (see Theorem~\ref{thm:convex}), its slope
$s(R)=\Delta^{*\prime}(R)$ is largest at the origin and decreases as
$R$ grows.
Write $s_0:=\Delta^{*\prime}(0^+)$ for the slope at the origin.
By the data-processing inequality \eqref{eq:dpi_diagonal}, the curve
lies weakly below the diagonal $\Delta\leq R$, so
\begin{align}
    s_0 \;\leq\; 1.
    \label{eq:s0_bound}
\end{align}
If $1/\beta>s_0$, the iso-cost lines are flatter than the curve at
the origin, no interior tangent exists, and the trivial encoder
$W\perp X$ is optimal.
If $1/\beta<s_0$, an interior tangent point exists. As $\beta$
grows it slides further along the curve.
The critical value is therefore
\begin{align}
    \beta_c \;:=\; \frac{1}{s_0} \;\geq\; 1,
    \label{eq:beta_c}
\end{align}
with equality $\beta_c=1$ when $H(W^*|Y)=0$, in which case the IB
curve is tangent to the diagonal at the origin and coincides with it
up to $R=I(X;Y)$.
For any $\beta<\beta_c$ we have $R(\beta)=0$ and
$\varepsilon(\beta)=I(X;Y)$. The interior of the IB curve is traced
only by $\beta\in(\beta_c,\infty)$.
Stochastic tasks typically have $\beta_c>1$. The exact value depends
on $p(x,y)$ through the slope of the IB curve at the origin, but
the qualitative picture is universal: phase transition at
$\beta_c\geq 1$, non-trivial representation above it, and
collapse below.

\begin{proposition}[Monotonicity in $\beta$]
    \label{prop:monotone_beta}
    For each $\beta\geq 0$, let $p_\beta(w|x)$ denote the optimal encoder
    minimizing $I(X;W)+\beta\,I(X;Y|W)$, and let
    $\varepsilon(\beta):=I(X;Y|W)$ and $R(\beta):=I(X;W)$ denote the
    distortion and rate achieved by $p_\beta$.
    Then $\varepsilon(\beta)$ is non-increasing in $\beta$ and $R(\beta)$
    is non-decreasing in $\beta$.
\end{proposition}

\begin{proof}
    See Appendix~\ref{app:monotone_beta}.
\end{proof}

\medskip
\noindent\textbf{Endpoint values.}
The two extremes of the IB curve are:
\begin{center}
\begin{tabular}{|l|l|l|l|l|l|}
\hline
Encoder & $\varepsilon$ & $R(\varepsilon)$ & $\beta(\varepsilon)$ & Slope & $c(\varepsilon)$ \\
\hline
Trivial & $I(X;Y)$ & $0$ & $\beta_c$ & $1/\beta_c$ & $0$ \\
Sufficient & $0$ & $H(W^*)$ & $\infty$ & $0$ & $I(X;Y)$ \\
\hline
\end{tabular}
\end{center}

\medskip
\noindent\textbf{Intercept $c(\varepsilon)$ is non-decreasing.}
As $\varepsilon$ decreases (lower distortion), $R(\varepsilon)$ increases and
$\beta(\varepsilon)$ increases by Proposition~\ref{prop:monotone_beta}.
Writing $\beta$ for $\beta(\varepsilon)$ and viewing $c$ as a function of $R$
(using $dR/d\varepsilon=-\beta$ and hence $d\varepsilon/dR=-1/\beta$),
\begin{align}
    \frac{dc}{dR}
    = \frac{1}{\beta} - \frac{d}{dR}\!\left(\frac{R}{\beta}\right)
    = \frac{1}{\beta} - \frac{\beta - R\,d\beta/dR}{\beta^2}
    = \frac{R}{\beta^2}\,\frac{d\beta}{dR} \geq 0,
    \label{eq:c_nondecreasing}
\end{align}
since $d\beta/dR\geq 0$ (see Proposition~\ref{prop:monotone_beta}) and $R\geq 0$.

\medskip
\noindent\textbf{IB curve regimes.}
The IB curve can be divided into three regimes:
\begin{itemize}
    \item $0\leq\beta\leq\beta_c$: the trivial encoder $W\perp X$ is
    optimal, giving the origin $(R,\Delta)=(0,0)$ and distortion
    $\varepsilon=I(X;Y)$.
    \item $\beta_c<\beta<\infty$: non-trivial stochastic soft
    clustering. The operating point
    $(R(\beta),\Delta(\beta))$ is determined by the
    self-consistency equations
    \eqref{eq:marginal}--\eqref{eq:decoder} and traces the interior
    of the IB curve.
    \item $\beta\to\infty$: hard assignment to the partition of
    Theorem~\ref{thm:mss}, recovering $W^*=p(Y|X)$ at the endpoint
    $(H(W^*),I(X;Y))$ by Corollary~\ref{cor:hard}.
\end{itemize}

\subsection{The Gap and Representation Regimes}

When $H(W^*|Y)>0$, the endpoint lies strictly below the diagonal:
the encoder uses $H(W^*)$ nats but only $I(X;Y)<H(W^*)$ are predictive of $Y$.
The gap $H(W^*|Y)$ is information in $W^*$ about which equivalence
class $x$ belongs to that is not deducible from $Y$ alone.
It vanishes whenever distinct equivalence classes $C_i\neq C_j$
assign probability to disjoint subsets of $\mathcal{Y}$ — equivalently,
whenever the map from $Y$ to $W^*$ is deterministic.
Otherwise the gap is strictly positive, an intrinsic cost of the
task's stochasticity.

\medskip
\noindent\textbf{Deterministic regime.}
If $Y$ is a deterministic function of $X$, then $H(Y|X)=0$ and
$W^*=p(Y|X)$ takes values at the vertices of $\mathcal{P}_K$.
Thus $H(W^*|Y)=0$, $H(W^*)=H(Y)=I(X;Y)$, and the endpoint lies on the diagonal.

\begin{example}[Deterministic binary task]
Let $X$ be uniform on $\{1,2,3,4\}$ with $Y=1$ for $X\in\{1,2\}$ and
$Y=0$ for $X\in\{3,4\}$.
$W^*$ takes two values $e_0=(1,0)$ and $e_1=(0,1)$, each with
probability $1/2$.
Since $H(Y|X)=0$, $I(X;Y)=H(Y)=\log 2\approx 0.693$ nats and
$H(W^*)=\log 2$. The endpoint $(\log 2,\log 2)$ lies on the diagonal.
By contrast, Example~\ref{ex:binary_example} has $H(W^*|Y)\approx 0.325$ nats
and the IB curve lies strictly below the converse boundary.
\end{example}

\medskip
\noindent\textbf{Near-deterministic regime.}
If $p(Y|X)$ is nearly deterministic, $H(Y|X)\approx 0$ and most
inputs map to predictive distributions near the vertices of $\mathcal{P}_K$.
Thus $I(X;Y)\approx H(Y)$, the map from $Y$ to $W^*$ is nearly
deterministic, and $H(W^*|Y)\approx 0$, so $H(W^*)\approx I(X;Y)$.
The gap is small and the IB curve nearly touches the converse boundary.

\medskip
\noindent\textbf{Ambiguous regime.}
If $p(Y|X)$ varies smoothly across the interior of the simplex,
$W^*=p(Y|X)$ takes values spread continuously inside $\mathcal{P}_K$.
$H(W^*)$ can then be much larger than $I(X;Y)$, the gap $H(W^*|Y)$
is large, and the IB curve bows significantly below the diagonal.

\medskip
The geometry of the optimal representation is most clearly described
at $\varepsilon=0$, where the magnitude of the gap determines the structure.
In all regimes, for $\varepsilon>0$ the operating point moves along the
IB curve toward the origin, with the encoder progressively merging
nearby regions of $\mathcal{M}$ into coarser clusters.

%%%%%%%%%%%%%%%%%%%%%%%%%%%%%%%%%%%%%%%%%%%%%%%%%%%%%%%%%%%%

\section{Learning with Less Supervision}
\label{sec:less_supervision}

\subsection{Beyond Supervised Learning}

In the fully supervised setting, the encoder is trained on i.i.d.\ samples
from the true joint distribution $p(x,y)$, directly optimizing for
sufficiency with respect to $p(Y|X)$.
In the semi-supervised setting \cite{vanengelen2020survey}, labeled
pairs $(x,y)$ are available but scarce or biased: they may come from
multiple sources, cover only a subset of the input space, or reflect a
task not representative of the full predictive structure of $X$.
In the self-supervised setting \cite{balestriero2025lejepa}, no human labels are
available during encoder training. The target $Y$ is replaced by a
self-generated signal, typically the embedding of another view of
the same input. Human labels are used only for downstream decoder
training.
In both semi-supervised and self-supervised settings, directly
optimizing for sufficiency with respect to the available labeling
risks overfitting to the observed signal.
A distributional regularizer on the encoder output is therefore needed.

\subsection{Maximum Entropy on the Predictive Manifold}

The canonical representation $W\in\mathcal{P}_K$ induces a measure on
the simplex determined jointly by $p(X)$ and the stochastic clustering
kernel (see Section~\ref{sec:geometric}).
When the labeling is biased, this measure reflects the bias of the
observed $(x,y)$ pairs rather than the true structure of $p(Y|X)$.

Maximum entropy on the induced measure over the simplex is the natural
response: it assumes no preference for any region of the predictive
manifold. All predictive distributions are treated equally.
This principle applies to the semi-supervised and self-supervised
settings, where labeling is absent or biased. In the fully supervised
setting, the conditional rate $I(X;W|Y)$ already constrains the encoder
toward the correct predictive structure (see Section~\ref{sec:supervised}).

The maximum entropy distribution on $\mathcal{P}_K$ with respect to the
Lebesgue measure is the flat Dirichlet \cite{cover2006elements}:
\begin{align}
    W \sim \mathrm{Dir}(1,\ldots,1),
    \label{eq:flat_dirichlet}
\end{align}
which assigns uniform density to all points in $\mathcal{P}_K$,
reflecting complete uncertainty about the labeling structure.
Direct optimization on the simplex is numerically problematic: gradients
vanish or explode near the boundary, and the reparameterization trick
is less stable than in Euclidean spaces.
A Euclidean relaxation of \eqref{eq:flat_dirichlet}, trading a
bounded entropy overhead for numerical tractability, is therefore
preferred.

\subsection{From Simplex to Euclidean Space}
\label{sec:simplex_to_euclidean}

The flat Dirichlet can be constructed from simpler distributions via
the following chain of transformations, each trading exact simplex
geometry for a looser but more tractable representation.

\medskip
\noindent\textbf{Step 1: Exponential construction of the flat Dirichlet.}
If $X_1,\ldots,X_K\stackrel{\mathrm{iid}}{\sim}\mathrm{Exp}(1/2)$, then
the normalized vector
\begin{align}
    D = \frac{(X_1,\ldots,X_K)}{\sum_{k=1}^K X_k}
    \label{eq:exp_construction}
\end{align}
follows $\mathrm{Dir}(1,\ldots,1)$ exactly \cite{ferguson1973bayesian}.
This construction replaces the $(K-1)$-dimensional simplex constraint
with $K$ non-negative dimensions.
The only extra degree of freedom introduced is the total scale
$S=\sum_k X_k$, which is divided out by the normalization in
\eqref{eq:exp_construction}.

\medskip
\noindent\textbf{Step 2: Gaussian construction of the exponential.}
If $Z_{k,1}, Z_{k,2}\stackrel{\mathrm{iid}}{\sim}\mathcal{N}(0,1)$,
then
\begin{align}
    X_k := Z_{k,1}^2 + Z_{k,2}^2 \sim \chi^2(2) = \mathrm{Exp}(1/2).
    \label{eq:gaussian_exp}
\end{align}
Thus the flat Dirichlet can be constructed from $2K$ i.i.d.\ standard
Gaussians.
Each Gaussian pair $(Z_{k,1},Z_{k,2})$ encodes two degrees of freedom:
the power $X_k=Z_{k,1}^2+Z_{k,2}^2$, which enters the construction,
and the phase $\theta_k=\arctan(Z_{k,2}/Z_{k,1})$, which does not.
The map to $X_k$ depends only on the power and is phase-invariant.

The three parameterizations are summarized in Table~\ref{tab:param}
and illustrated in Figure~\ref{fig:chain}.
Each trades exactness for tractability:
the flat Dirichlet is exact but numerically unstable near the simplex
boundary \cite{joo2020dirichlet},
the exponential parameterization introduces a scale overhead of
$\tfrac{1}{2}\log(2\pi eK)$ nats, vanishing relative to the total
entropy $K$ as $K\to\infty$, and
the Gaussian parameterization incurs a non-diminishing phase entropy
overhead of $K\log 2\pi$ nats but provides an unconstrained, fully
reparameterizable optimization landscape.
\begin{table}[h]
\centering
\begin{tabular}{lccl}
\hline
\textbf{Prior} & \textbf{Dimensions} & \textbf{Extra DoF} &
\textbf{Entropy overhead} \\
\hline
$\mathrm{Dir}(1,\ldots,1)$ & $K-1$ & None & None \\
$K$ i.i.d.\ $\mathrm{Exp}(1/2)$ & $K$ & 1 scale & $\tfrac{1}{2}\log(2\pi eK)$ nats \\
$2K$ i.i.d.\ $\mathcal{N}(0,1)$ & $2K$ & 1 scale $+$ $K$ phases &
$\tfrac{1}{2}\log(2\pi eK) + K\log(2\pi)$ nats \\
\hline
\end{tabular}
\caption{Parameterizations of the flat Dirichlet and their entropy overheads.}
\label{tab:param}
\end{table}

\begin{figure}[h]
\centering
\begin{tikzpicture}[
    node distance=2.2cm,
    box/.style={draw, rounded corners, minimum height=1.8cm,
    minimum width=2.8cm, align=center, fill=blue!5},
    arrow/.style={->, thick, >=stealth}
]
\node[box] (gauss) {$2K$ i.i.d.\\
    $Z_{k,j}\sim\mathcal{N}(0,1)$\\[2pt]
    \footnotesize dim $= 2K$};

\node[box, right=of gauss] (exp) {$K$ i.i.d.\\
    $X_k\sim\mathrm{Exp}(1/2)$\\[2pt]
    \footnotesize dim $= K$};

\node[box, right=of exp] (dir) {Flat Dirichlet\\
    $D\sim\mathrm{Dir}(1,\ldots,1)$\\[2pt]
    \footnotesize dim $= K{-}1$};

\draw[arrow] (gauss) -- (exp);
    % node[midway, above, font=\footnotesize, align=center] {$X_k = Z_{k,1}^2+Z_{k,2}^2$\\[1pt] \footnotesize loses $K$ phases}
    % node[midway, below, font=\footnotesize, align=center] {overhead: $K\log 2\pi$ nats};
    
\draw[arrow] (exp) -- (dir);
    % node[midway, above, font=\footnotesize, align=center] {$D_k = X_k/\!\sum_j X_j$\\[1pt] \footnotesize loses 1 scale}
    % node[midway, below, font=\footnotesize, align=center] {overhead: $\tfrac{1}{2}\!\log(2\pi eK)$ nats};

\node[font=\footnotesize, gray, below=0.3cm of gauss] {unconstrained, smooth};
\node[font=\footnotesize, gray, below=0.3cm of exp] {non-negative};
\node[font=\footnotesize, gray, below=0.3cm of dir] {simplex-constrained};
\end{tikzpicture}
\caption{The chain of transformations connecting the isotropic
Gaussian to the flat Dirichlet on $\mathcal{P}_K$.
Reading right to left: the flat Dirichlet is the maximum entropy
prior on $\mathcal{P}_K$. The exponential parameterization embeds
it in $\mathbb{R}_+^K$ with a diminishing scale overhead. The
Gaussian parameterization embeds it in $\mathbb{R}^{2K}$ with a
non-diminishing phase overhead of $K\log 2\pi$ nats in $I(X;W)$
that is invariant under the simplex map, but provides the smoothest
optimization landscape.}
\label{fig:chain}
\end{figure}
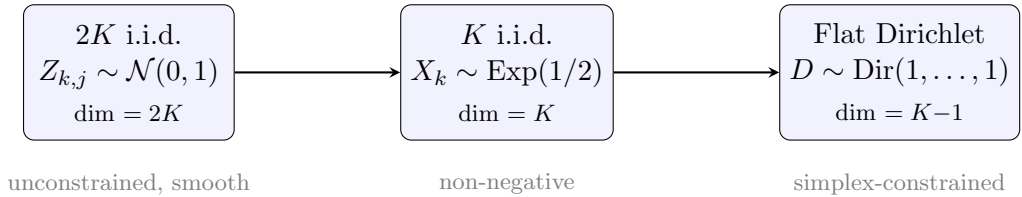

\subsection{Connection to SIGReg}
\label{sec:sig_connection}

SIGReg \cite{balestriero2025lejepa}
is a Gaussian relaxation of the maximum entropy principle on the
predictive manifold.
SIGReg enforces full distributional matching of the encoder marginal
to $\mathcal{N}(0,I)$, not merely moment matching.
Given a minibatch of $N$ encoder outputs
$w_1,\ldots,w_N\in\mathbb{R}^{2K}$, SIGReg projects them along $M$
random unit directions $a_1,\ldots,a_M$ and applies the Epps-Pulley
characteristic-function test \cite{epps1983test} to each set of
one-dimensional projections $\{a_m^\top w_n\}_{n=1}^N$, then
averages the test statistics:
\begin{align}
    \mathcal{L}_{\mathrm{SIGReg}}(q(w))
    := \frac{1}{M}\sum_{m=1}^{M}
    T_{\mathrm{EP}}\!\bigl(\{a_m^\top w_n\}_{n=1}^N\bigr),
    \label{eq:sigreg}
\end{align}
where $q(w) := \mathbb{E}_{p(x)}[q(w|x)]$ is the encoder marginal
estimated from the minibatch, and $T_{\mathrm{EP}}$ is the Epps-Pulley statistic testing
normality of a univariate sample.
By the Cram\'er-Wold theorem \cite{billingsley1995probability}, matching all one-dimensional projections
to a standard Gaussian is equivalent to matching the full distribution
to $\mathcal{N}(0,I)$.
This is strictly stronger than moment matching: matching any finite
number of moments does not determine the distribution in general
\cite{balestriero2025lejepa}, but characteristic-function matching does.

SIGReg enforces the $2K$ Gaussian parameterization of the flat Dirichlet.
The encoder outputs live in $\mathbb{R}^{2K}$, not on the simplex.
The simplex and the Gaussian space are alternative representations of
the same predictive information, related by the chain of
Section~\ref{sec:simplex_to_euclidean}.
No explicit inverse to the simplex is needed in practice: the decoder
maps directly from the Gaussian representation to predictions.
This parameterization introduces a non-diminishing phase entropy
overhead of $K\log 2\pi$ nats relative to the exact flat Dirichlet.
This overhead is a cost in \emph{rate accounting only}, not in
achievable prediction.
By the construction of Section~\ref{sec:simplex_to_euclidean}, the map
$\pi:\mathbb{R}^{2K}\to\mathcal{P}_K$ defined by
$\pi(w)_k := (w_{k,1}^2+w_{k,2}^2)/\sum_j(w_{j,1}^2+w_{j,2}^2)$
depends on $w$ only through the squared radii
$r_k^2=w_{k,1}^2+w_{k,2}^2$. It is therefore invariant under the
action of $SO(2)^K$ on the $K$ phase planes.
The phases are nuisance coordinates of the simplex map: they carry
no predictive content.

Concretely, applying the data processing inequality to the chain
$X\to W\to\pi(W)$ gives
\begin{align}
    I(\pi(W);Y) \;\leq\; I(W;Y) \;\leq\; I(X;Y),
    \label{eq:phase_dpi}
\end{align}
Any decoder that depends on $w$ only through $\pi(w)$ achieves
$I(\pi(W);Y)=I(W;Y)$.
Equivalently, when $W$ is generated by the Dirichlet--Gaussian chain
the Bayes-optimal predictive distribution $p(y|w)$ is itself
phase-invariant.
The $K\log 2\pi$ phase entropy is therefore rate that the encoder
\emph{can} spend on phase-equivalent distinctions contributing nothing
to $I(W;Y)$: a horizontal shift of the operating point in the
$(R,\Delta)$ plane, increasing $I(X;W)$ without reducing $I(W;Y)$.

The exponential parameterization eliminates the phase slack, reducing
the overhead to the vanishing $\tfrac{1}{2}\log(2\pi eK)$ scale term,
at the cost of non-negativity constraints and less stable gradients
near the simplex boundary \cite{joo2020dirichlet}.
The Gaussian parameterization is preferred when $\beta$ is a free
hyperparameter: the phase overhead is then absorbed into $\beta$
without affecting predictions.
The exponential parameterization is more appropriate when the exact
value of $I(X;W)$ matters, for example when $\beta$ is fixed
by a constraint or when estimating the IB curve quantitatively.

SIGReg has three consequences within the IB framework.
First, it prevents representational collapse: the encoder cannot map
all inputs to a constant, since this would violate the isotropic
Gaussian constraint.
Second, since the marginal $p(W)$ is constrained to approximate
$\mathcal{N}(0,I)$, the marginal entropy $H(W)$ is approximately
constant. The encoding rate
\begin{align}
    I(X;W) = H(W) - H(W|X)
    \label{eq:rate_decomp}
\end{align}
is controlled primarily via $H(W|X)$: compression is achieved by
increasing the conditional entropy of $W$ given $X$, rather than by
collapsing the marginal.
Third, the encoder can be trained independently of the decoder and of
the specific labeling: \eqref{eq:sigreg} requires no labels and can be
applied to all available data, labeled or not.

%%%%%%%%%%%%%%%%%%%%%%%%%%%%%%%%%%%%%%%%%%%%%%%%%%%%%%%%%%%%

\section{Practical Learning using IB}

Optimizing \eqref{eq:opt} directly requires access to true mutual
information quantities, which are intractable in general.
Three learning settings are considered.
In the supervised case, the IB Lagrangian is estimated from
minibatches via the CEB decomposition.
In the semi-supervised case, labels are scarce: the conditional rate
is restricted to labeled data and SIGReg regularizes the encoder on
all data.
In the self-supervised case, no labels are available. The CEB
conditional rate is replaced by a view-prediction proxy and SIGReg
serves as the distributional regularizer.
In all three cases, the encoder is the primary object of interest.
Training proceeds in two stages.

\medskip
\noindent\textbf{Stage 1: Encoder training.}
The encoder $q_\phi(w|x)$ is trained by minimizing
\begin{align}
    \mathcal{L}_{\mathrm{enc}}(\phi) :=
    \mathcal{L}_{\mathrm{pred}}(\phi) +
    \mathcal{L}_{\mathrm{shape}}(\phi),
    \label{eq:enc_loss}
\end{align}
where $\mathcal{L}_{\mathrm{pred}}$ measures how well the representation
predicts the target, and $\mathcal{L}_{\mathrm{shape}}$ 
regularizes the encoder distribution $q_\phi(w|x)$.
The specific forms of both terms, and the justification for each choice,
are given in the subsections below.

\medskip
\noindent\textbf{Stage 2: Decoder training.}
With the encoder frozen, the decoder $q_\psi(y|w)$ is trained by
minimizing the cross-entropy loss on labeled data:
\begin{align}
    \mathcal{L}_{\mathrm{dec}}(\psi) :=
    -\mathbb{E}_{p(x,y)}\mathbb{E}_{q_\phi(w|x)}
    [\log q_\psi(y|w)].
    \label{eq:dec_loss}
\end{align}
Minimizing this over $\psi$ drives $q_\psi(y|w)$ toward the optimal
decoder $p(y|w)$: by the standard cross-entropy decomposition
\cite{bishop2006pattern}, $\mathcal{L}_{\mathrm{dec}}(\psi) =
\mathbb{E}_{q_\phi(w)}[D_{\mathrm{KL}}(p(y|w)\|q_\psi(y|w))] +
H(Y|W)$, so the unique minimizer is $q_\psi(y|w)=p(y|w)$.

\subsection{Supervised Learning}
\label{sec:supervised}

In the fully supervised setting, labeled pairs $(x,y)$ are available
for all training examples.
Two formulations of the IB objective \eqref{eq:opt} are considered:
the CEB of Fischer
\cite{fischer2020conditional} and the VIB of Alemi et al.\ \cite{alemi2017deep}.
Both use variational bounds for practical optimization in their
original papers.
Here the CEB objective is estimated via minibatch class-conditional
marginals rather than a learned backward encoder, avoiding an
additional variational approximation.
Standard cross-entropy training of the encoder optimizes
$\mathcal{L}_{\mathrm{pred}}$ alone, which upper-bounds the
distortion $I(X;Y|W)$ \cite{alemi2017deep}, with no
$\mathcal{L}_{\mathrm{shape}}$ term. The encoder is not penalized for
retaining redundant information about $X$.

\subsubsection{CEB: Conditional Entropy Bottleneck}
\label{sec:ceb}

By the chain rule for mutual information
(see Appendix~\ref{app:lagrangian_deriv}), the IB Lagrangian of
\eqref{eq:opt} is equivalent, up to the constant $\beta\,I(X;Y)$, to
\begin{align}
    \mathcal{L}_{\mathrm{enc}}^{\text{\textsc{ceb}}}(\phi)
    := \underbrace{\beta\,I(X;W|Y)}_{\mathcal{L}_{\mathrm{pred}}}
    + \underbrace{(1-\beta)\,I(X;W)}_{\mathcal{L}_{\mathrm{shape}}}.
    \label{eq:lagrangian_prime}
\end{align}
$\mathcal{L}_{\mathrm{pred}} = \beta\,I(X;W|Y)$ is the
\textit{conditional rate} \cite{wyner1976rate,fischer2020conditional}: the
information $W$ retains about $X$ beyond what $Y$ already reveals,
minimized when $W$ is a sufficient statistic for $Y$.
$\mathcal{L}_{\mathrm{shape}} = (1-\beta)\,I(X;W)$ is the encoding rate:
for $\beta\geq\beta_c\geq 1$ its coefficient is negative, so it acts as a
regularizer that maximizes the encoding rate, driving the encoder away
from the trivial solution while the prediction term enforces
sufficiency.
The objective is bounded since $I(X;W)\leq H(X)$.

Fischer \cite{fischer2020conditional} optimizes \eqref{eq:lagrangian_prime}
via variational bounds, introducing a learned backward encoder
$b(w|y)$ to approximate $p(w|y)$.
Here both terms are instead estimated from minibatch class-conditional
marginals, requiring no additional learned approximation.
The conditional rate decomposes as
\begin{align}
    I(X;W|Y) = \mathbb{E}_{p(x,y)}\left[D_{\mathrm{KL}}
    \bigl(q_\phi(w|x)\,\|\,q_\phi(w|y)\bigr)\right],
    \label{eq:ixw_given_y}
\end{align}
where
\begin{align}
    q_\phi(w|y) := \mathbb{E}_{p(x|y)}[q_\phi(w|x)]
    \label{eq:qwy}
\end{align}
is the class-conditional marginal of the encoder, estimated per
minibatch as
\begin{align}
    \hat{q}_\phi(w|y) := \frac{1}{N_y}\sum_{x_i:\,y_i=y} q_\phi(w|x_i),
    \label{eq:qwy_hat}
\end{align}
where $N_y$ is the number of labeled examples with label $y$ in the
minibatch. This is an unbiased empirical estimate of \eqref{eq:qwy}
under the i.i.d.\ training data assumption.
The total rate satisfies
\begin{align}
    I(X;W) = \mathbb{E}_{p(x)}\left[D_{\mathrm{KL}}
    \bigl(q_\phi(w|x)\,\|\,q_\phi(w)\bigr)\right],
    \label{eq:compression}
\end{align}
where $q_\phi(w):=\mathbb{E}_{p(x)}[q_\phi(w|x)]$ is estimated from
the minibatch as a mixture.
Evaluating the KL requires computing
$\mathbb{E}_{q_\phi(w|x)}[\log q_\phi(w)]$. Since $q_\phi(w)$ is a
finite mixture, this expectation has no closed form
\cite{bishop2006pattern} and is estimated by drawing samples $w\sim q_\phi(w|x)$ and
evaluating $\log q_\phi(w)$ at each sample.

The CEB supervised encoder loss is therefore
\begin{align}
    \mathcal{L}_{\mathrm{enc}}^{\text{\textsc{ceb}}}(\phi)
    = \underbrace{\beta\,\mathbb{E}_{p(x,y)}\left[D_{\mathrm{KL}}
    \bigl(q_\phi(w|x)\,\|\,q_\phi(w|y)\bigr)\right]}_{\mathcal{L}_{\mathrm{pred}}}
    + \underbrace{(1-\beta)\,\mathbb{E}_{p(x)}\left[D_{\mathrm{KL}}
    \bigl(q_\phi(w|x)\,\|\,q_\phi(w)\bigr)\right]}_{\mathcal{L}_{\mathrm{shape}}}.
    \label{eq:sup_loss}
\end{align}

\subsubsection{VIB: Variational Information Bottleneck}
\label{sec:vib}

VIB \cite{alemi2017deep} uses a Gaussian encoder
$q_\phi(w|x)=\mathcal{N}(\mu_\phi(x),\mathrm{diag}\,\sigma^2_\phi(x))$
with $W\in\mathbb{R}^d$ and a fixed prior $a(w):=\mathcal{N}(0,I)$ to upper-bound $I(X;W)$ as
\begin{align}
    I(X;W) \;\leq\;
    \mathbb{E}_{p(x)}\bigl[D_{\mathrm{KL}}(q_\phi(w|x)\,\|\,a(w))\bigr].
    \label{eq:vib_kl}
\end{align}
The VIB encoder and decoder are trained jointly using the
loss\footnote{We follow the convention of
\cite{tishby1999information}, with $\beta$ on the prediction term.
\cite{alemi2017deep} place $\beta$ on the KL term, related by
$\beta_{\text{Alemi}}=1/\beta$.}
\begin{align}
    \mathcal{L}^{\text{\textsc{vib}}}(\phi,\psi)
    := \underbrace{\beta(-\,\mathbb{E}_{p(x,y)}\mathbb{E}_{q_\phi(w|x)}
    [\log q_\psi(y|w)])}_{\mathcal{L}_{\mathrm{pred}}} 
    + \underbrace{\mathbb{E}_{p(x)}\bigl[D_{\mathrm{KL}}
    (q_\phi(w|x)\,\|\,a(w))\bigr]}_{\mathcal{L}_{\mathrm{shape}}},
    \label{eq:vib_loss}
\end{align}
where $\mathcal{L}_{\mathrm{pred}}$ takes the same form as the cross-entropy decoder
loss \eqref{eq:dec_loss}.

CEB and VIB differ in how they estimate the same IB objective.
CEB uses minibatch marginals, targeting the true mutual information
without a fixed prior. The cost is sampling to evaluate the mixture
KL.
VIB uses closed-form KL bounds, but the
fixed $\mathcal{N}(0,I)$ prior biases the encoder, making
$\mathcal{L}_{\mathrm{shape}}$ an upper bound on $I(X;W)$ with the gap
acting as an uncontrolled regularization beyond the IB objective.
Tschannen et al.\ \cite{tschannen2020mutual} show that such
prior-induced regularization can dominate performance, making empirical
gains from VIB difficult to attribute to compression alone.

\subsection{Semi-Supervised Learning with SIGReg}
\label{sec:semi_supervised}

In the semi-supervised setting \cite{grandvalet2004semi,kingma2014semi},
labeled pairs $(x,y)$ are scarce while unlabeled data is abundant.
Two failure modes arise: training on labeled data alone yields a biased
encoder, as the labeled subset rarely covers the full support of $p(Y|X)$.
Training on unlabeled data without a sufficiency constraint yields the
trivial solution.
Meaningful training on all available data therefore requires a
regularizer that prevents collapse while remaining agnostic to the
missing labels.

In the absence of label information, the least biased choice is the
maximum entropy distribution on $\mathcal{P}_K$ \eqref{eq:flat_dirichlet},
which imposes no preference over predictive classes.
SIGReg \eqref{eq:sigreg} implements this prior as a distributional
regularizer on the encoder output (see Section~\ref{sec:sig_connection}).
The semi-supervised encoder loss is
\begin{align}
    \mathcal{L}_{\mathrm{enc}}^{\mathrm{semi}}(\phi) :=
    \underbrace{\mathbb{E}_{q(x,y)}\left[D_{\mathrm{KL}}
    \bigl(q_\phi(w|x)\,\|\,q_\phi(w|y)\bigr)\right]}_{\mathcal{L}_{\mathrm{pred}}\;\text{(labeled data only)}}
    + \underbrace{\lambda\,\mathcal{L}_{\mathrm{SIGReg}}(q_\phi(w))}_{\mathcal{L}_{\mathrm{shape}}}.
    \label{eq:semi_loss}
\end{align}
The first term targets the conditional rate $I(X;W|Y)$, estimated from
the labeled distribution $q(x,y)$, and encourages the encoder to discard non-predictive
information.
The second term applies SIGReg to all data, preventing collapse.
The trade-off coefficient $\lambda\geq 0$ balances the two:
too small and the encoder overfits the labeled subset,
too large and the distributional prior dominates,
washing out class-discriminative structure.

\subsection{Self-Supervised Learning with SIGReg}
\label{sec:self_supervised}

In the self-supervised setting, no labels are available
during encoder training.
Without any sufficiency signal, minimizing $I(X;W)$ alone yields the
trivial solution.
The sufficiency signal is instead provided by a view-prediction loss \cite{blum1998combining,chen2020simple}:
given augmented views $\{x_v\}$ of the same input, the encoder
is trained so that the embeddings $\{f_\phi(x_v)\}$ are close to
a joint global view, replacing the conditional rate $I(X;W|Y)$ with a proxy
based on shared semantic content \cite{balestriero2025lejepa}.

The self-supervised encoder loss is
\begin{align}
    \mathcal{L}_{\mathrm{enc}}^{\mathrm{self}}(\phi) :=
    \underbrace{\mathbb{E}_{p(x)}\bigl[
        D_{\mathrm{V}}(\{q_\phi(w|x_v)\})
    \bigr]}_{\mathcal{L}_{\mathrm{pred}}}
    + \underbrace{\lambda\,\mathcal{L}_{\mathrm{SIGReg}}(q_\phi(w))}_{\mathcal{L}_{\mathrm{shape}}},
    \label{eq:self_loss}
\end{align}
where $D_{\mathrm{V}}$ penalizes spread among the view embeddings \cite{balestriero2025lejepa}, and $\mathcal{L}_{\mathrm{SIGReg}}$ prevents collapse
by enforcing the distributional prior on all data
(see Section~\ref{sec:sig_connection}).
The trade-off coefficient $\lambda\geq 0$ balances prediction against
regularization: too small and the encoder ignores distributional
structure. Too large and the prior dominates, washing out
view-discriminative content.
Since no labels are used, the encoder training is arguably task-agnostic.

\section{Experiments}
\label{sec:experiments}

Four experiments are presented.
The first two are toy problems where the ground-truth predictive
manifold $\mathcal{M}=\{p(Y|x):x\in\mathcal{X}\}$ can be visualised
on the probability simplex $\mathcal{P}_3$. The encoder is a Dirichlet
distribution $W\sim\mathrm{Dir}(\alpha(x))$ and $\beta$ is swept to
trace the IB curve.
The third applies the minibatch-marginal CEB estimator to FashionMNIST
and compares it with the standard VIB prior bound.
The fourth ablates over the simplex dimension~$K$, confirming that it
is the only structural parameter and is easy to set.
Every $\beta$ sweep uses the identical grid
$\beta\in\{0.5,5,10,25,50,100,250\}$, with $\beta=0.5$ included as a
deliberate sub-critical probe: by the phase transition
\eqref{eq:beta_c}, $\beta<1$ lies at or below $\beta_c$ and should
drive the encoder toward the trivial solution, which is verified
empirically in all three settings.
All reported accuracy, rate, and distortion values in
Sections~\ref{sec:exp_mnist}--\ref{sec:exp_k_ablation} are from a
single training run per configuration (seed $42$).

\subsection{Continuous input}
\label{sec:exp_cont}

The input is an angle $X=\theta\sim\mathrm{Uniform}[0,2\pi)$ and $Y\in\{0,1,2\}$.  The ground-truth conditional
$p(Y=y|x)\propto\exp(\cos(x-2\pi y/3))$ traces a smooth
\emph{loop} through the simplex (see Figure~\ref{fig:cont_manifold}).
Because $p(Y|x)$ is injective and $x$ is continuous, the rate
$I(X;W^*)=\infty$, so the IB curve extends to infinite rate.

The standard IB Lagrangian
$\mathcal{L}=I(X;W)+\beta\,I(X;Y|W)$ is minimised using the known
$p(Y|X)$.
As $\beta$ increases, the Dirichlet means progressively recover the
full predictive loop (see Figure~\ref{fig:cont_sweep}). The
rate--distortion trade-off is shown in Figure~\ref{fig:cont_info}.

\begin{figure}[h]
\centering
\includegraphics[width=0.35\textwidth]{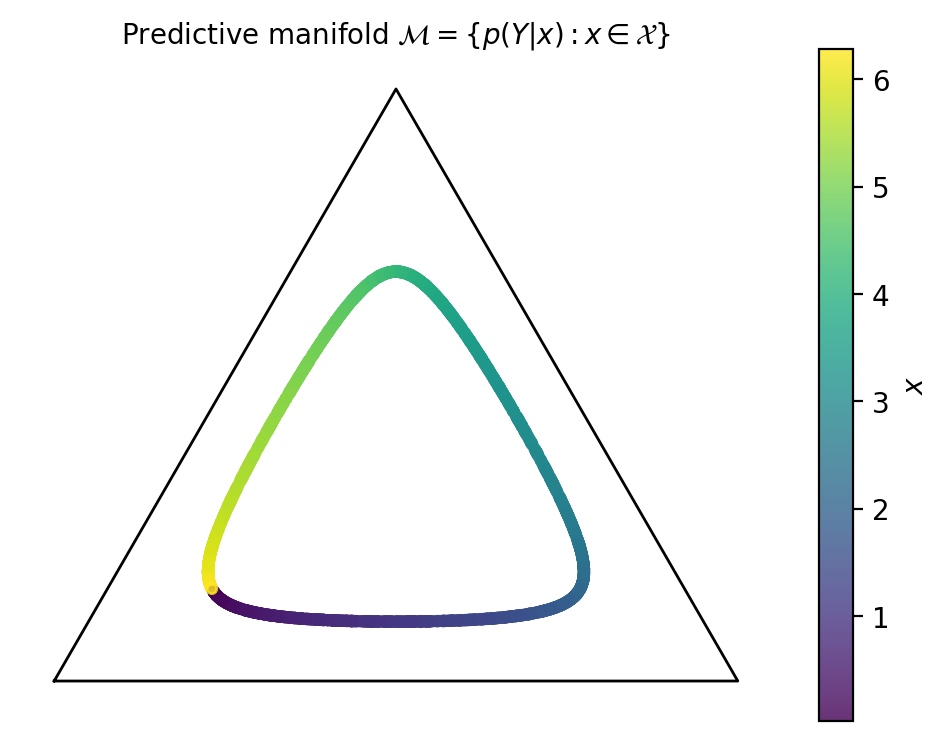}
\caption{Continuous predictive manifold: $p(Y|x)$ traces a loop through
$\mathcal{P}_3$ as $\theta$ varies over $[0,2\pi)$.}
\label{fig:cont_manifold}
\end{figure}

\begin{figure}[h]
\centering
\includegraphics[width=\textwidth]{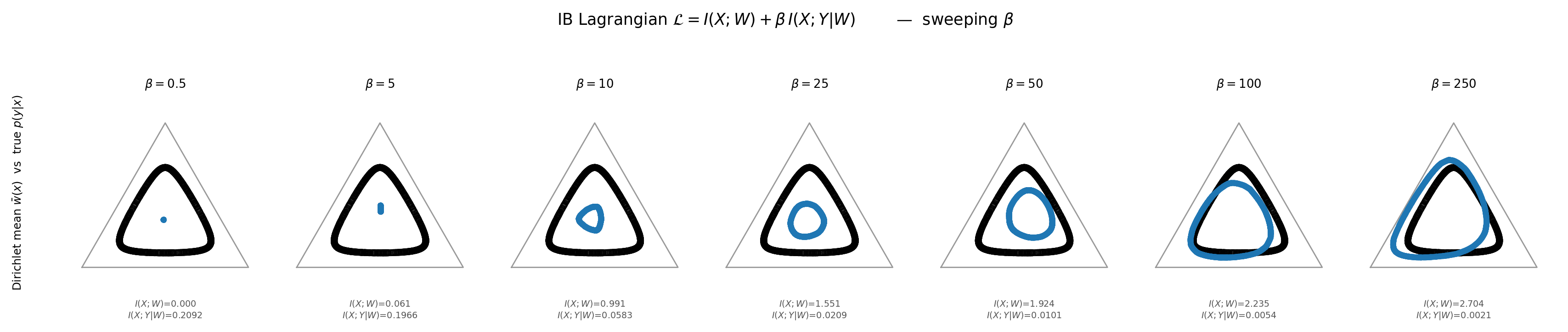}
\caption{Continuous experiment: Dirichlet means (blue) vs.\ ground-truth
$p(Y|x)$ (black circles) for increasing $\beta$.}
\label{fig:cont_sweep}
\end{figure}

\begin{figure}[h]
\centering
\includegraphics[width=\textwidth]{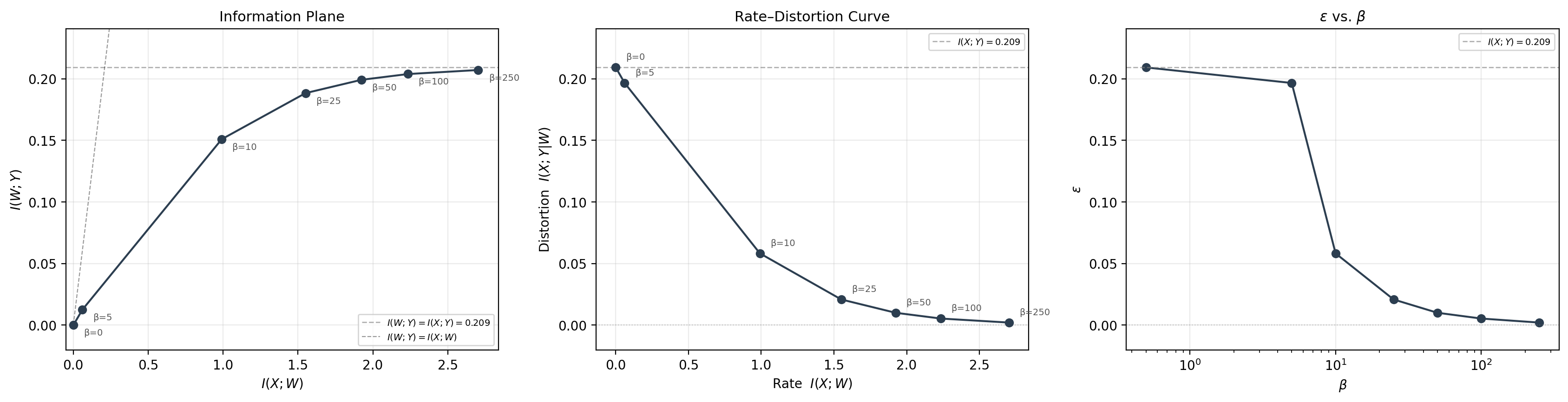}
\caption{Continuous experiment: Information Plane, Rate--Distortion curve,
and $\varepsilon(\beta)$.}
\label{fig:cont_info}
\end{figure}

\subsection{Discrete input}
\label{sec:exp_disc}

Here $X\in\{0,\ldots,19\}$ is uniform, with pairs of consecutive values
sharing the same conditional ($Y=\lfloor X/2\rfloor$ gives 10 clusters).
The 10 distributions $p(Y|x)$ are placed at equally-spaced angles on the
same circular contour used in Section~\ref{sec:exp_cont} (see Figure~\ref{fig:disc_manifold}), giving
$H(X)=\log 20\approx 3.00$ but $H(W^*)=\log 10\approx 2.30$.

The CEB Lagrangian
$\mathcal{L}=(1{-}\beta)\,I(X;W)+\beta\,I(X;W|Y)$ is estimated
non-parametrically via leave-one-out (LOO) log-mixture estimators.
No knowledge of $p(Y|X)$ enters the loss.
As $\beta$ grows the encoder separates the 10 clusters
(see Figure~\ref{fig:disc_sweep}). The rate saturates near
$H(W^*)=\log 10$, well below $H(X)$ (see Figure~\ref{fig:disc_info}):
the encoder discards the redundant bit that distinguishes the two
$X$-values within each cluster.

The learned Dirichlet means do not coincide with the ground-truth
simplex points. They form a different partition of the simplex that
is nonetheless \emph{equivalent} for prediction.
The IB objective only requires $W$ to be a sufficient statistic for $Y$.
Any injective map from the 10 clusters into $\mathcal{P}_3$ achieves
$I(X;Y|W)=0$ and is therefore a global optimum.
The particular embedding is a gauge freedom of the loss
(see Remark~\ref{rem:gauge_decoder}), so different random seeds
(or $\beta$ values) land on different but equally valid configurations.

\begin{figure}[h]
\centering
\includegraphics[width=0.35\textwidth]{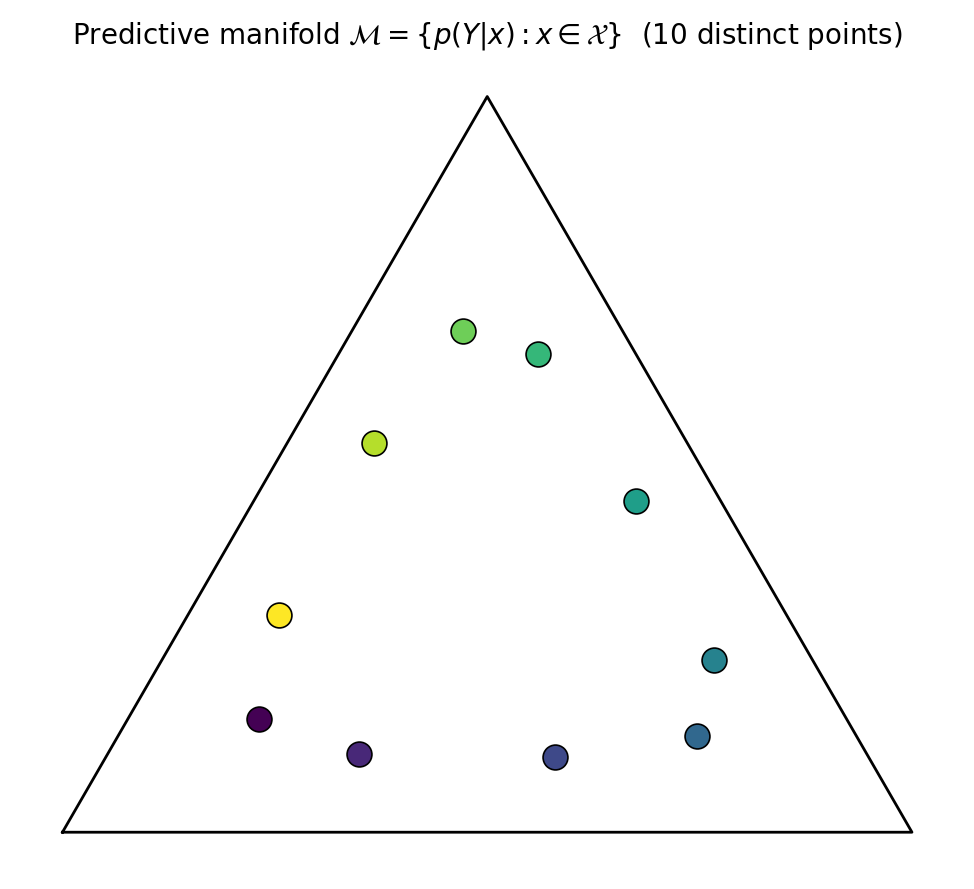}
\caption{Discrete predictive manifold: 10 distinct $p(Y|x)$ points on
the simplex loop.}
\label{fig:disc_manifold}
\end{figure}

\begin{figure}[h]
\centering
\includegraphics[width=\textwidth]{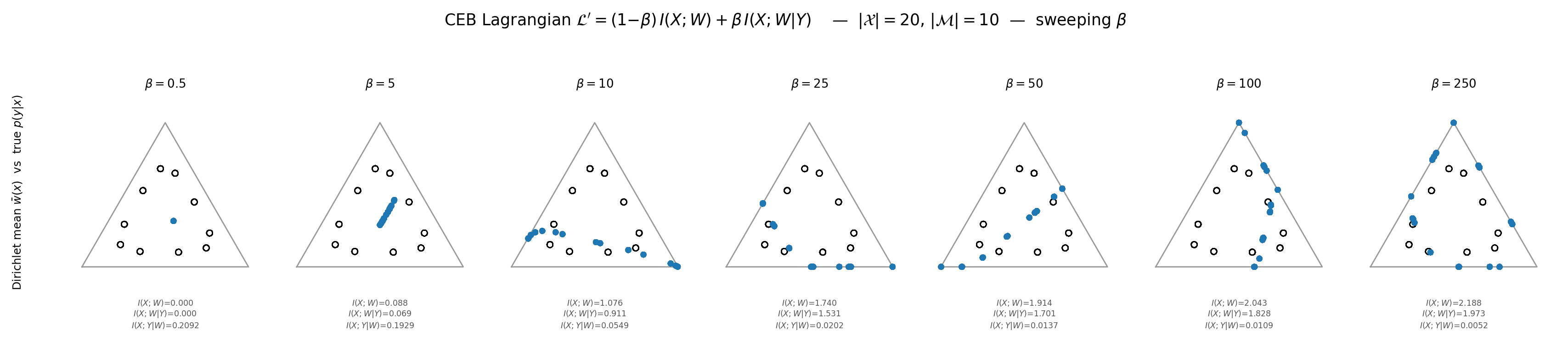}
\caption{Discrete CEB experiment: Dirichlet means vs.\ ground-truth
for increasing $\beta$.}
\label{fig:disc_sweep}
\end{figure}

\begin{figure}[h]
\centering
\includegraphics[width=\textwidth]{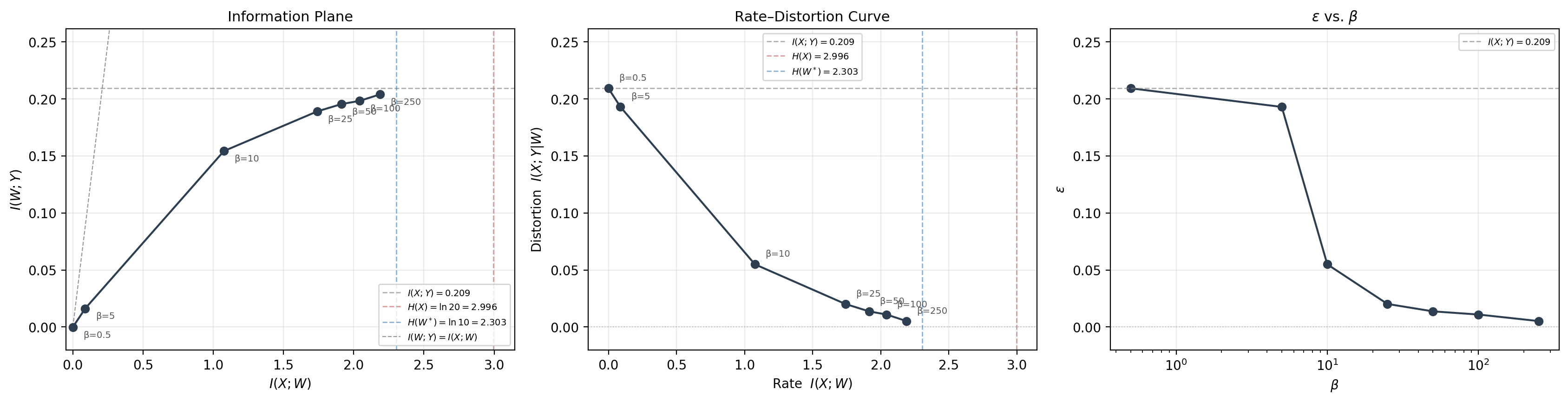}
\caption{Discrete CEB experiment: Information Plane, Rate--Distortion
curve, and $\varepsilon(\beta)$.  Dashed lines mark $H(X)$ and $H(W^*)$.}
\label{fig:disc_info}
\end{figure}

\subsection{VIB vs.\ CEB}
\label{sec:exp_mnist}

The standard VIB~\cite{alemi2017deep} is compared with the CEB formulation on
FashionMNIST ($28{\times}28$ grayscale, 10 classes), sweeping
$\beta\in\{0.5,5,10,25,50,100,250\}$.
Both methods share the same convolutional backbone. They differ in
the latent space, the loss, and the assumptions required.

\textbf{VIB} uses a Gaussian encoder
$q_\phi(w|x)=\mathcal{N}(\mu_\phi(x),\mathrm{diag}\,\sigma^2_\phi(x))$
with $W\in\mathbb{R}^{16}$ and a learned multilayer perceptron (MLP)
decoder $q_\psi(y|w)$.  The loss is
\begin{align}
    \mathcal{L}^{\text{\textsc{vib}}}(\phi)
    = \underbrace{\beta\bigl(-\,\mathbb{E}_{p(x,y)}\mathbb{E}_{q_\phi(w|x)}
    [\log q_\psi(y|w)]\bigr)}_{\mathcal{L}_{\mathrm{pred}}}
    + \underbrace{\mathbb{E}_{p(x)}\bigl[D_{\mathrm{KL}}
    \bigl(q_\phi(w|x)\,\big\|\,a(w)\bigr)\bigr]}_{\mathcal{L}_{\mathrm{shape}}},
\end{align}
where $a(w)=\mathcal{N}(0,I)$, cf.\ \eqref{eq:vib_loss}.

\textbf{CEB} uses a Dirichlet encoder
$q_\phi(w|x)=\mathrm{Dir}(\alpha_\phi(x))$ with
$W\in\mathcal{P}_{10}$, the 10-class probability simplex.
The loss is~\eqref{eq:sup_loss},
\begin{align}
    \mathcal{L}_{\mathrm{enc}}^{\text{\textsc{ceb}}}(\phi)
    = \underbrace{\beta\,\mathbb{E}_{p(x,y)}\!\left[D_{\mathrm{KL}}
    \bigl(q_\phi(w|x)\,\|\,q_\phi(w|y)\bigr)\right]}_{\mathcal{L}_{\mathrm{pred}}}
    + \underbrace{(1{-}\beta)\,\mathbb{E}_{p(x)}\!\left[D_{\mathrm{KL}}
    \bigl(q_\phi(w|x)\,\|\,q_\phi(w)\bigr)\right]}_{\mathcal{L}_{\mathrm{shape}}},
\end{align}
where $q_\phi(w|y)$ and $q_\phi(w)$ are the class-conditional and
marginal distributions estimated from minibatch LOO log-mixture
estimators.
Unlike VIB, this requires no fixed prior, no decoder during
training, and no variational bound. The loss targets the true mutual
information directly.
The \emph{only free parameter} beyond the shared backbone is the
number of components of $W$ (here $K{=}10$, matching the number of
classes).

Because the CEB loss is purely information-theoretic, the encoder
exhibits gauge freedom (see Remark~\ref{rem:gauge_decoder}): the learned
simplex coordinates need not align with class labels.
Classification accuracy is evaluated by training a small post-hoc
readout MLP ($10\to 256\to 10$) on the frozen encoder means, separating representation learning from
label alignment.

Figure~\ref{fig:mnist_plots} shows the results.
CEB matches or exceeds VIB at every operating point from $\beta{=}5$
onward, peaking at ${\sim}91.4\%$ accuracy ($\beta{=}25$) versus
VIB's best of ${\sim}90.7\%$ ($\beta{=}250$), while maintaining a
consistently lower rate.
For a fair rate comparison, $I(X;W)$ is reported for both methods
using the same minibatch LOO log-mixture estimator applied to each
encoder's own distribution family (Dirichlet for CEB, Gaussian for
VIB). VIB's variational KL-to-prior upper bound is not used, as it
is not comparable to CEB's non-parametric estimate.

\medskip
\noindent\textbf{Relation to the phase transition.}
The sweep $\beta\in\{0.5,5,10,25,50,100,250\}$ is designed to span
the critical value $\beta_c$ of \eqref{eq:beta_c}.
FashionMNIST is a near-deterministic classification task: $Y$ is
almost a function of $X$, so $H(W^*|Y)\approx 0$ and
$\beta_c\approx 1$.
Consistent with this, $\beta=0.5$ lies at or below the phase
transition and CEB performs poorly there. $\beta\geq 5$ sits
comfortably above $\beta_c$ and yields the non-trivial operating
points reported above.
Two mechanisms prevent the minibatch-marginal estimator from
collapsing at usable $\beta$.
First, the LOO log-mixture estimator of $I(X;W)$ uses an
$(N{-}1)$-component mixture rather than the true marginal, biasing
the estimated compression cost and softening the pull toward the
trivial point.
Second, the stochastic Dirichlet encoder acts as implicit
regularization, keeping $H(W|X)$ bounded away from zero and
preventing degenerate deterministic solutions.

\begin{figure}[h]
\centering
\includegraphics[width=\textwidth]{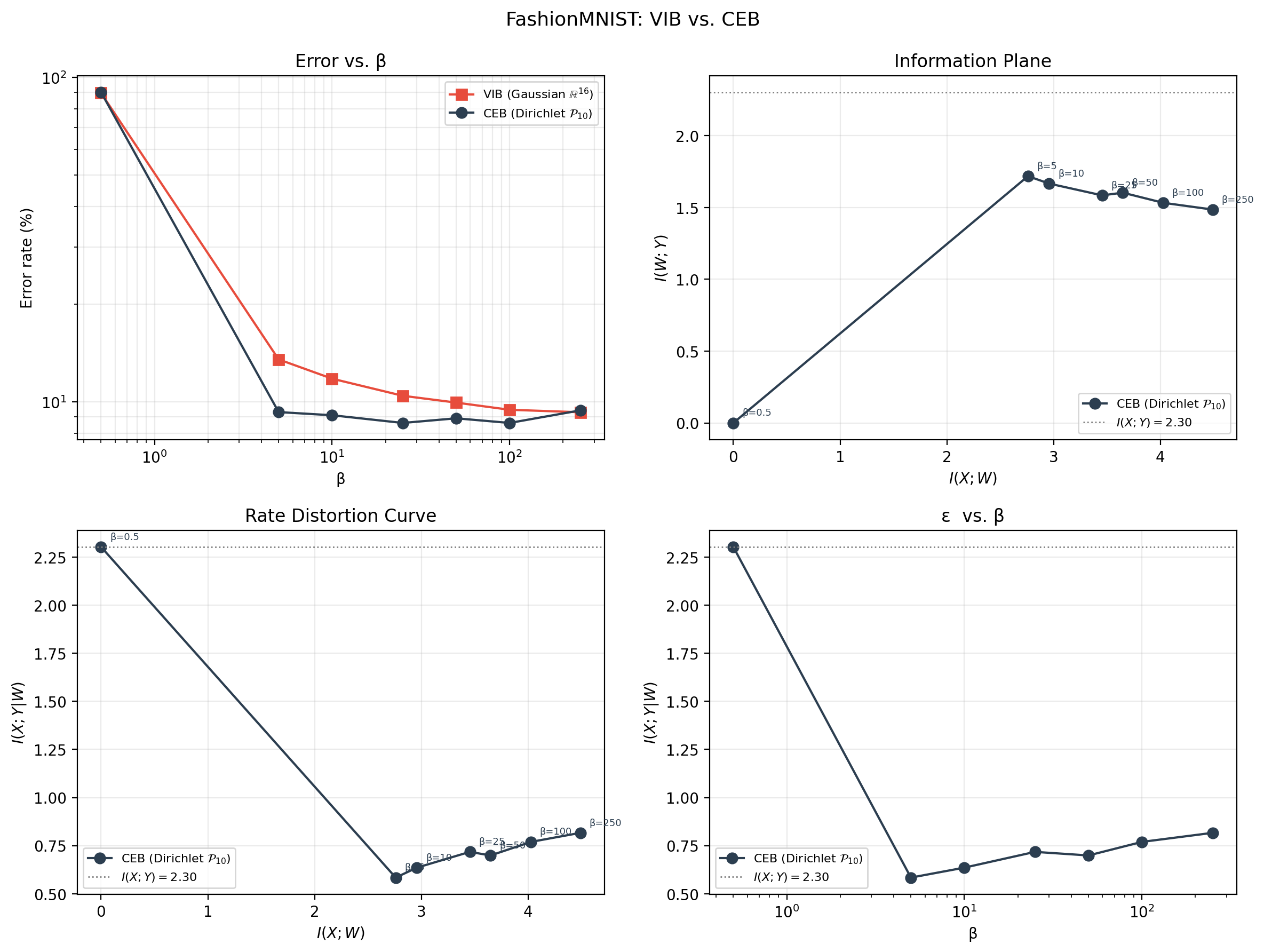}
\caption{FashionMNIST: VIB (Gaussian, $\mathbb{R}^{16}$) vs.\ CEB
(Dirichlet, $\mathcal{P}_{10}$).
\emph{Top-left:} error rate vs.\ $\beta$.
\emph{Top-right:} information plane $I(W;Y)$ vs.\ $I(X;W)$ (CEB only).
\emph{Bottom-left:} rate--distortion curve $I(X;Y|W)$ vs.\ $I(X;W)$ (CEB only).
\emph{Bottom-right:} distortion $\varepsilon=I(X;Y|W)$ vs.\ $\beta$ (CEB only).
CEB achieves higher accuracy at lower rate with no parametric assumptions
beyond the dimension of $W$.}
\label{fig:mnist_plots}
\end{figure}

\subsection{Ablation over the simplex dimension $K$}
\label{sec:exp_k_ablation}

The CEB encoder has a single structural parameter: the simplex
dimension~$K$.  For a deterministic classification task such as
FashionMNIST, the true predictive manifold
$\mathcal{M}=\{p(Y|x):x\in\mathcal{X}\}$ consists of $|\mathcal{Y}|=10$
point masses (the simplex vertices $e_0,\ldots,e_9$), so the effective
degrees of freedom equal the number of classes.
We therefore expect the encoder to need at least $K=|\mathcal{Y}|$
components and to gain little from $K>|\mathcal{Y}|$.

$\beta=25$ (the best CEB operating point from
Section~\ref{sec:exp_mnist}) is fixed and $K\in\{3,5,10,15,20\}$
is swept.
All other settings match Section~\ref{sec:exp_mnist}: the same
convolutional backbone, batch size ($2048$), optimizer (Adam,
$\mathrm{lr}{=}10^{-3}$), encoder training ($90$ epochs), and
post-hoc readout ($30$ epochs, same architecture). $K$ is the only
quantity that varies.
Figure~\ref{fig:k_ablation} shows the results.
Accuracy rises from $K{=}3$ ($88.7\%$) to $K{=}10$ ($90.9\%$)
and plateaus for $K>10$ ($91.2\%$ at $K{=}15$ and $K{=}20$).
Distortion $I(X;Y|W)$ continues to decrease with $K$, reflecting the
additional capacity, but the predictive benefit saturates once $K$
reaches the number of classes.
$K=|\mathcal{Y}|$ is a natural default: sufficient and near-optimal.

\medskip
\noindent\textbf{Behaviour for $K<|\mathcal{Y}|$.}
The only sub-$|\mathcal{Y}|$ point in the sweep is $K{=}3$, which
gives $88.7\%$ accuracy, a noticeable drop of roughly $2$\,pp
from the $K{=}10$ value of $90.9\%$, but not a total collapse.
This is consistent with our framing: once $K$ falls below the number
of classes, multiple predictive clusters must share simplex vertices
and some class information is necessarily lost, but the remaining
$K$ vertices can still carry most of the label signal, so the
degradation is gradual rather than catastrophic.
$K{=}1$ and $K{=}2$ are excluded from the sweep as their behaviour
is determined a priori: $\mathcal{P}_1$ is a single point, so
$W\perp X$ by construction and accuracy is exactly chance ($10\%$).
$\mathcal{P}_2$ supports at most one non-trivial split and cannot
separate ten classes.
The trend $K{=}3\to 5\to 10$ and the plateau at $K>|\mathcal{Y}|$
together confirm that $K=|\mathcal{Y}|$ is the natural default.

\begin{figure}[h]
\centering
\includegraphics[width=\textwidth]{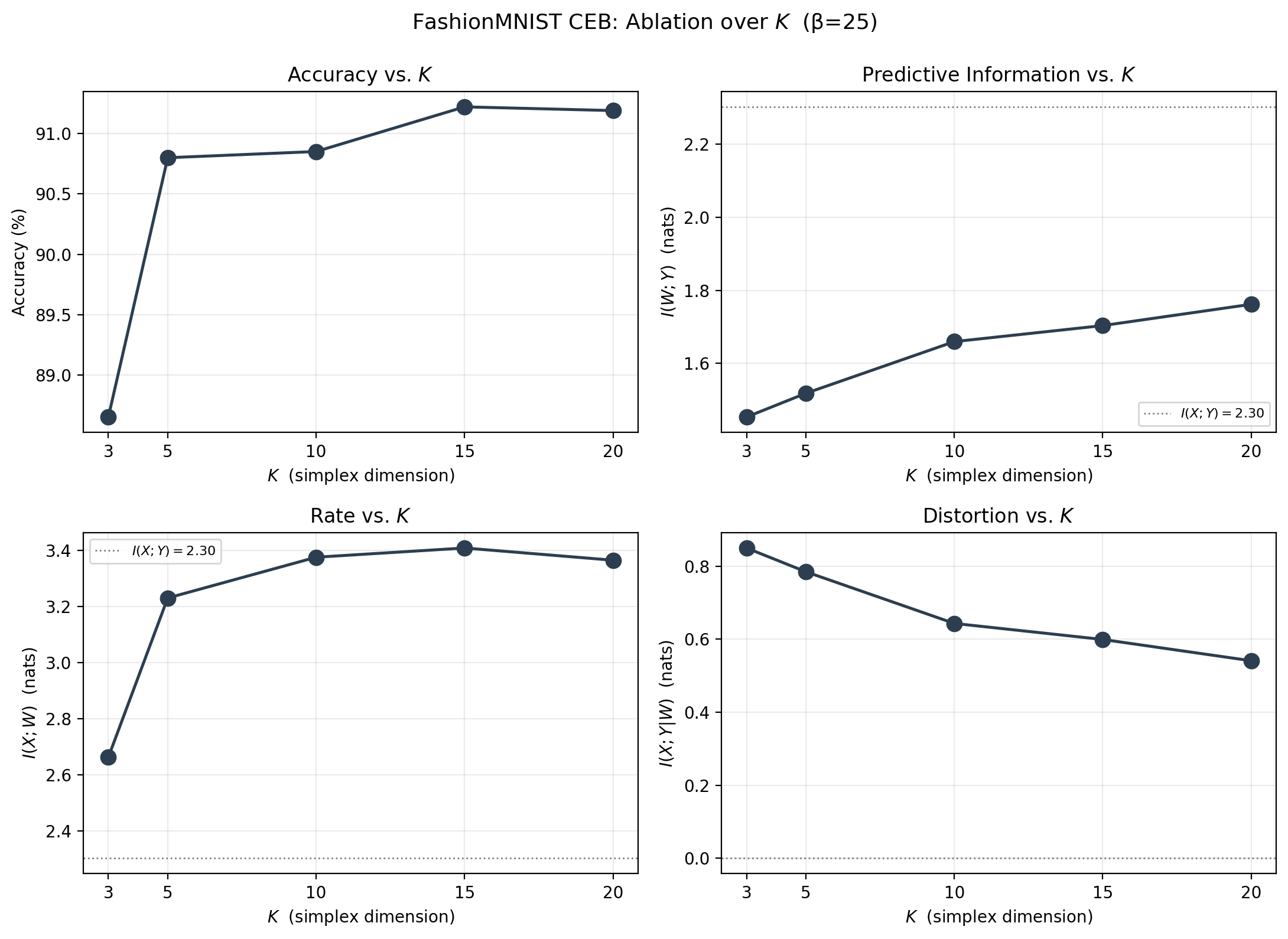}
\caption{CEB ablation over simplex dimension $K$ at $\beta=25$.
Accuracy saturates at $K{=}|\mathcal{Y}|{=}10$.
Further increasing $K$ reduces distortion but does not improve
classification.}
\label{fig:k_ablation}
\end{figure}

%%%%%%%%%%%%%%%%%%%%%%%%%%%%%%%%%%%%%%%%%%%%%%%%%%%%%%%%%%%%

\section{Conclusion}

We presented a geometric framework for representation learning
grounded in the Information Bottleneck, showing that the optimal
encoder at any distortion level performs a soft clustering of the
predictive manifold $\mathcal{M}\subset\mathcal{P}_K$.
The Dirichlet-to-Gaussian chain connects the maximum entropy prior on
the simplex to the isotropic Gaussian enforced by SIGReg, identifying
it as a Gaussian relaxation with a non-diminishing entropy overhead
that affects rate accounting but not achievable prediction.
This relaxation provides a principled distributional regularizer for
semi-supervised and self-supervised learning.
The Conditional Entropy Bottleneck decomposition yields practical
encoder losses for supervised and semi-supervised settings. In the
self-supervised setting, SIGReg serves as the regularizer alongside a
view-prediction proxy.

Several directions remain open.
The continuous-$Y$ bound (see Proposition~\ref{prop:geom_dim}) relies on
asymptotic covering-number dimension. Replacing this with a more
operational notion would strengthen the result.
The exponential parameterization offers a near-exact alternative to
the Gaussian with a vanishing overhead, and its practical viability
deserves empirical investigation.
Finally, scaling the non-parametric CEB estimator to larger datasets
and higher-dimensional tasks would test the limits of the minibatch
marginal approach.

%%%%%%%%%%%%%%%%%%%%%%%%%%%%%%%%%%%%%%%%%%%%%%%%%%%%%%%%%%%%

\bibliographystyle{plain}
\bibliography{main}

%%%%%%%%%%%%%%%%%%%%%%%%%%%%%%%%%%%%%%%%%%%%%%%%%%%%%%%%%%%%

\appendix

\section*{Appendices}

\section{Effective Dimension in the Continuous Case}
\label{app:eff_dim}

When $Y$ is continuous, $p(Y|x)$ is a density and the predictive
manifold $\mathcal{M}=\{p(Y|x):x\in\mathcal{X}\}$ is a subset of an
infinite-dimensional function space.\footnote{We continue to call
$\mathcal{M}$ a ``manifold'' loosely: it is the image of
$x\mapsto p(Y|x)$ and may have singularities, self-intersections, or
varying dimension. Only the metric-space structure $(\mathcal{M},d)$
is used in what follows.}
The effective dimension of $\mathcal{M}$ is finite under mild regularity
conditions on $f^*\colon x\mapsto p(Y|x)$ and the input space
$\mathcal{X}$.

Let $(S,d)$ be a metric space, where $d:S\times S\to\mathbb{R}_{\geq 0}$
is a distance function.
The \textit{covering number} $N(\delta,S,d)$ is the smallest number of
$d$-balls of radius $\delta$ needed to cover $S$.
The \textit{effective dimension} of $S$ \cite{kolmogorov1961epsilon}
is the exponent governing the growth of the covering
number,\footnote{The term \textit{effective dimension} has been used
in several disjoint senses in the literature. Here we use it as our
name for this covering-number-based variant (equivalently, the upper
box-counting / Minkowski dimension), in the spirit of
Kolmogorov--Tikhomirov's $\varepsilon$-entropy
\cite{kolmogorov1961epsilon}.}
\begin{align}
    d_{\mathrm{eff}}(S) := \lim_{\delta\to 0}
    \frac{\log N(\delta,\,S,\,d)}{\log(1/\delta)},
    \label{eq:eff_dim_def}
\end{align}
when this limit exists (or the corresponding $\limsup$).

\begin{proposition}[Geometric bound {\cite[\S 3.2]{falconer2003fractal}}]
\label{prop:geom_dim}
Suppose the support of $p(X)$ has intrinsic dimension $d_X$, meaning
$\log N(\delta,\,\mathrm{supp}(p_X),\,\|\cdot\|)
= \Theta(d_X\log(1/\delta))$.
If $f^*$ is $L$-Lipschitz with respect to a metric $\rho$ on the space
of distributions (e.g.\ Hellinger distance $d_H$ or total variation),
\begin{align}
    \rho\bigl(p(Y|x_1),\,p(Y|x_2)\bigr)
    \leq L\,\|x_1-x_2\|,
\end{align}
then $d_{\mathrm{eff}}(\mathcal{M})\leq d_X$.
\end{proposition}

\begin{proof}
Any $(\delta/L)$-cover of $\mathrm{supp}(p_X)$ maps under $f^*$ to a
$\delta$-cover of $\mathcal{M}$ in $\rho$.
Hence
\begin{align}
    \log N(\delta,\,\mathcal{M},\,\rho)
    \;\leq\;
    \log N\!\bigl(\delta/L,\,\mathrm{supp}(p_X),\,\|\cdot\|\bigr)
    \;=\;
    \Theta\!\bigl(d_X\log(1/\delta)\bigr),
\end{align}
so $d_{\mathrm{eff}}(\mathcal{M})\leq d_X$.
\end{proof}

The intrinsic dimension $d_X$ reflects the manifold hypothesis
\cite{fefferman2016testing,pope2021intrinsic}: even when $X$ lives in a
high-dimensional ambient space, $p(X)$ typically concentrates near a
lower-dimensional submanifold.
Lipschitz regularity of $f^*$ ensures this low dimension transfers to
the predictive manifold $\mathcal{M}$.

\begin{example}[Scalar Gaussian channel]
Let $Y=X+Z$ with $Z\sim\mathcal{N}(0,\sigma_Z^2)$ independent of $X$,
so $p(Y|x)=\mathcal{N}(x,\sigma_Z^2)$.
The predictive manifold is
$\mathcal{M}=\{\mathcal{N}(x,\sigma_Z^2):x\in\mathbb{R}\}$,
a one-parameter family (varying mean, fixed variance).
Since $X$ is scalar, $d_X=1$.
To verify the Lipschitz condition, note that
$d_H^2(p(Y|x_1),p(Y|x_2))=1-\exp(-(x_1{-}x_2)^2/8\sigma_Z^2)
\leq(x_1{-}x_2)^2/8\sigma_Z^2$,
so $d_H(p(Y|x_1),p(Y|x_2))\leq|x_1{-}x_2|/(2\sqrt{2}\,\sigma_Z)$:
$f^*$ is Lipschitz with constant $L=1/(2\sqrt{2}\,\sigma_Z)$.
Proposition~\ref{prop:geom_dim} gives $d_{\mathrm{eff}}(\mathcal{M})\leq d_X=1$,
hence $K=\lceil d_{\mathrm{eff}}(\mathcal{M})\rceil+1=2$ and $\mathcal{P}_2$ suffices.
Since $p(Y|x)$ is injective, $\mathcal{M}$ is a genuine 1D curve and
$\mathcal{P}_1$ (a single point) is insufficient: $K=2$ is tight.
\end{example}

\begin{remark}
For discrete $Y$, $K=|\mathcal{Y}|$ exactly. The joint bound is
$K\leq\min\bigl(|\mathcal{Y}|,\,\lceil d_X\rceil+1\bigr)$.
For continuous $Y$, the optimal representation admits a
reparameterization of dimension at most $d_X$ in general, since $X$
itself is sufficient.
The stronger statement $d_{\mathrm{eff}}(\mathcal{M})\leq d_X$,
and hence $K\leq\lceil d_X\rceil+1$ as defined, requires Lipschitz
regularity of $f^*$ (see Proposition~\ref{prop:geom_dim}).
\end{remark}

\section{Proof of Theorem~\ref{thm:convex}}
\label{app:convexity}

\begin{proof}
The argument is the standard time-sharing construction for convexity
of the rate-distortion function \cite[Theorem~10.2.1]{cover2006elements},
adapted to the IB distortion measure.
Let $\varepsilon_1,\varepsilon_2\in[0,I(X;Y)]$ and
$\lambda\in[0,1]$.
For any $\delta>0$, let $p_1(w|x)$ be a feasible encoder with
$I_1(X;Y|W)\leq\varepsilon_1$ and
$I_1(X;W)\leq R(\varepsilon_1)+\delta$, and
let $p_2(w|x)$ be a feasible encoder with
$I_2(X;Y|W)\leq\varepsilon_2$ and
$I_2(X;W)\leq R(\varepsilon_2)+\delta$.
Such encoders exist. For finite alphabets the minimum is achieved,
so one may take $\delta=0$.

Let $Q\sim\mathrm{Bernoulli}(\lambda)$ be independent of $(X,Y)$,
and define
\begin{align}
    W := \begin{cases}
    (1,W_1) & \text{if } Q=1, \\
    (0,W_2) & \text{if } Q=0,
    \end{cases}
\end{align}
where $W_i\sim p_i(w|x)$.
Since each $p_i(w|x)$ satisfies $Y-X-W_i$ and $Q\perp(X,Y)$, the
combined encoder satisfies $Y-X-W$: for $w=(q,w_q)$,
\begin{align}
    p(w|x,y) = p(q|x,y)\,p_q(w_q|q,x,y) = p(q|x)\,p_q(w_q|q,x) = p(w|x).
\end{align}

\textit{Rate decomposition.}
Since $Q\perp X$,
\begin{align}
    I(X;W) &= I(X;Q,W_Q) = I(X;Q) + I(X;W_Q|Q) \notag\\
    &= 0 + \lambda\,I_1(X;W_1) + (1-\lambda)\,I_2(X;W_2) \notag\\
    &\leq \lambda\,(R(\varepsilon_1)+\delta)
    + (1-\lambda)\,(R(\varepsilon_2)+\delta) \notag\\
    &= \lambda\,R(\varepsilon_1) + (1-\lambda)\,R(\varepsilon_2)
    + \delta.
    \label{eq:rate_ts}
\end{align}

\textit{Distortion decomposition.}
Since $W=(Q,W_Q)$,
\begin{align}
    I(X;Y|W) &= I(X;Y|Q,W_Q) \notag\\
    &= \lambda\,I_1(X;Y|W_1) + (1-\lambda)\,I_2(X;Y|W_2) \notag\\
    &\leq \lambda\,\varepsilon_1 + (1-\lambda)\,\varepsilon_2.
    \label{eq:dist_ts}
\end{align}

The time-shared encoder is feasible at distortion
$\bar{\varepsilon}=\lambda\varepsilon_1+(1-\lambda)\varepsilon_2$.
Its rate is at most $\lambda R(\varepsilon_1)+(1-\lambda)R(\varepsilon_2)+\delta$.
Since $R(\bar{\varepsilon})$ is the minimum rate at distortion
$\bar{\varepsilon}$,
\begin{align}
    R\bigl(\lambda\varepsilon_1+(1-\lambda)\varepsilon_2\bigr)
    \leq \lambda\,R(\varepsilon_1)+(1-\lambda)\,R(\varepsilon_2)
    +\delta.
\end{align}
Taking $\delta\to 0$ gives
\begin{align}
    R\bigl(\lambda\varepsilon_1+(1-\lambda)\varepsilon_2\bigr)
    \leq \lambda\,R(\varepsilon_1)+(1-\lambda)\,R(\varepsilon_2).
\end{align}

\textit{Monotonicity.}
The feasible set $\{p(w|x):I(X;Y|W)\leq\varepsilon\}$ grows with
$\varepsilon$, so the minimum rate cannot increase.

\textit{Concavity of $\Delta^*(R)$.}
Write $\Delta^*(R)=I(X;Y)-\varepsilon^*(R)$, where $\varepsilon^*(R)$
is the inverse of $R(\varepsilon)$.
Since $R(\varepsilon)$ is convex and non-increasing, $\varepsilon^*(R)$
is convex and non-increasing, so $\Delta^*(R)$ is concave and
non-decreasing.
\end{proof}

\section{Flat Portion of the IB Curve}
\label{app:flat_portion}

For any $R\in[H(W^*),H(X)]$, set $W=(W^*,U)$ with $U=g(X)$ for a
suitably chosen function $g$.
Since $W^*=f^*(X)$ is a function of $X$, the Markov chain $Y-X-W$ is
preserved.
The rate $I(X;W)\geq H(W^*)$ increases with the information $g$
retains about $X$, reaching $I(X;W)=H(X)$ at $g(X)=X$.

For the predictive information, the data processing inequality on
$W\to W^*$ gives $I(W;Y)\geq I(W^*;Y)=I(X;Y)$. The data processing
inequality on $Y-X-W$ gives $I(W;Y)\leq I(X;Y)$, so
$I(W;Y)=I(X;Y)$.
The IB curve is therefore flat at $I(X;Y)$ for all $R\in[H(W^*),H(X)]$.

\section{Proof of Proposition~\ref{prop:monotone_beta}}
\label{app:monotone_beta}

\begin{proof}
Let $\beta_1<\beta_2$ and let $p_1,p_2$ denote the corresponding
optimal encoders achieving $(R_1,\varepsilon_1)$ and
$(R_2,\varepsilon_2)$ respectively.
By optimality of $p_1$ at $\beta_1$,
\begin{align}
    R_1 + \beta_1\,\varepsilon_1 \leq R_2 + \beta_1\,\varepsilon_2.
\end{align}
By optimality of $p_2$ at $\beta_2$,
\begin{align}
    R_2 + \beta_2\,\varepsilon_2 \leq R_1 + \beta_2\,\varepsilon_1.
\end{align}
Adding these two inequalities gives
\begin{align}
    (\beta_2-\beta_1)\,\varepsilon_2
    \leq (\beta_2-\beta_1)\,\varepsilon_1.
\end{align}
Since $\beta_2>\beta_1$, dividing by $\beta_2-\beta_1>0$ yields
$\varepsilon_2\leq\varepsilon_1$, so $\varepsilon(\beta)$ is
non-increasing.

For the rate, the first optimality inequality gives
\begin{align}
    R_1 - R_2 \leq \beta_1\,(\varepsilon_2 - \varepsilon_1),
\end{align}
and the second gives
\begin{align}
    R_2 - R_1 \leq \beta_2\,(\varepsilon_1 - \varepsilon_2),
\end{align}
equivalently $R_1 - R_2 \geq \beta_2\,(\varepsilon_2-\varepsilon_1)$.
Since $\varepsilon_2-\varepsilon_1\leq 0$ and $\beta_1\geq 0$, the
first inequality gives
$R_1-R_2\leq\beta_1(\varepsilon_2-\varepsilon_1)\leq 0$, so
$R_1\leq R_2$ and $R(\beta)$ is non-decreasing.
\end{proof}

\section{Derivation of the Equivalent Lagrangian}
\label{app:lagrangian_deriv}

Expanding $I(X;W,Y)=I(X;W)+I(X;Y|W)=I(X;Y)+I(X;W|Y)$ in two ways
gives the chain-rule identity
\begin{align}
    I(X;Y|W) + I(X;W) = I(X;W|Y) + I(X;Y).
    \label{eq:cmi_identity}
\end{align}
Since $I(X;Y)$ is constant with respect to the encoder, the
Lagrangian $\mathcal{L}=I(X;W)+\beta\,I(X;Y|W)$ can be rewritten as
\begin{align}
    \mathcal{L}
    &= I(X;W) + \beta\,I(X;Y|W) \notag\\
    &= I(X;W) + \beta\bigl(I(X;W|Y) + I(X;Y) - I(X;W)\bigr) \notag\\
    &= (1-\beta)\,I(X;W) + \beta\,I(X;W|Y) + \beta\,I(X;Y).
    \label{eq:lagrangian_equiv}
\end{align}
Dropping the constant $\beta\,I(X;Y)$ yields the equivalent
Lagrangian \eqref{eq:lagrangian_prime}.

\end{document}